\documentclass[11pt]{article}

\def\final{1}

\usepackage[colorlinks,linkcolor=myblue,citecolor=mygreen, hypertexnames = false]{hyperref}
\usepackage{amsmath,amsfonts, mathdots, amssymb, latexsym, amsbsy,bm,amsthm,mathtools} 
\usepackage{braket}
\usepackage[margin=1in]{geometry}
\usepackage[english]{babel}
\usepackage[utf8]{inputenc}
\usepackage{xcolor}

\colorlet{myblue}{blue!80!black}
\colorlet{mygreen}{green!40!black}

\usepackage{xspace}

\usepackage{paralist}

\usepackage{times}

\usepackage{framed}

\usepackage{bbm}
\usepackage{nicefrac}

\usepackage{thm-restate}

\usepackage[nameinlink]{cleveref}
\crefname{lemma}{Lemma}{Lemmas}
\crefname{thm}{Theorem}{Theorems}

\newtheorem{thm}{Theorem}[section]
\newtheorem{theorem}[thm]{Theorem}
\newtheorem{lemma}[thm]{Lemma}

\newtheorem{claim}[thm]{Claim}

\theoremstyle{definition}

\newtheorem{definition}[thm]{Definition}

\newenvironment{nestedproof}{\begin{proof}}{\end{proof}}

\ifdefined\final
\usepackage{microtype}
\newcommand{\neil}[1]{}
\newcommand{\nnote}[1]{}
\newcommand{\laura}[1]{}
\newcommand{\leon}[1]{}

\newcommand{\modified}[1]{#1}

\newcommand{\todo}[1]{}
\newcommand{\TODO}[1]{}
\newcommand{\suggestion}[1]{ }

\else

\usepackage[colorinlistoftodos]{todonotes}
\let\TODO\todo

\renewcommand{\todo}[1]{{\color{red!50!black}\em \small [TODO: #1]}}

\definecolor{lightblue}{rgb}{0.38,0.82,0.90}
\colorlet{lightyellow}{yellow!60!white}
\newcommand{\neil}[1]{\TODO[color=lightyellow,inline]{\textbf{Neil:} #1}}
\newcommand{\nnote}[1]{{\color{yellow!50!black}\em \small [Neil: #1]}}
\newcommand{\laura}[1]{{\color{lightblue!70!black}\em \small [Laura: #1]}}
\newcommand{\leon}[1]{{\color{blue!50!black}\em \small [Leon: #1]}}

\newcommand{\modified}[1]{{\color{purple!80!black}#1}}

\newcommand{\suggestion}[1]{{\color{green!40!black}\em \small [Suggestion: #1]}}
\fi

\usepackage{graphicx}
\usepackage[font=small,labelfont=bf]{caption}
\usepackage{subcaption}

\graphicspath{{figures/}}

\usepackage{authblk}

\newcommand{\R}{\mathbb{R}}

\newcommand{\N}{\mathbb{Z}_{\geq 0}}

\newcommand{\Z}{\mathbb{Z}}
\newcommand{\Znonneg}{\Z_{\geq 0}}

\newcommand{\abs}[1]{\left\lvert #1 \right\rvert}
\newcommand{\norm}[1]{\left\lVert #1 \right\rVert}
\newcommand{\OPT}{\ensuremath{\mathsf{OPT}}\xspace}

\newcommand{\lcirc}{l^\circ}
\newcommand{\linit}{l^\circ}
\newcommand{\lclose}{l^\star}
\newcommand{\lcloseup}{\bar{l}^\star}
\newcommand{\emphdef}[1]{\emph{#1}}

\newcommand{\lambdamin}{\lambda_{\textrm{min}}}
\newcommand{\dep}{d}
\renewcommand{\d}{\dep}
\newcommand{\quickdep}{\tilde{d}}

\renewcommand{\epsilon}{\varepsilon}

\newcommand{\dist}{\mathrm{d}}

\renewcommand{\l}{\ell} 
\newcommand{\capmin}{\nu_{\min}} 
\newcommand{\capsum}{\nu_{\Sigma}} 
\renewcommand{\b}{b} 

\renewcommand{\a}{a} 
\renewcommand{\b}{b} 
\renewcommand{\c}{c}

\newcommand{\Iloc}{I}
\newcommand{\T}{\overline{T}}

\newcommand{\Einf}{E^{\infty}}
\newcommand{\Etilde}{\tilde{E}}
\newcommand{\Eloc}{\Etilde}

\newcommand{\ind}{\mathbbm{1}}

\newcommand{\CCL}{\cite{cominetti2015existence}}

\newcommand{\TrajLip}{\kappa}

\newcommand{\centerdist}{\Gamma}
\newcommand{\centercompat}{r_{\!\mathrm{c}}}

\newcommand{\Tglob}{T_{\mathrm{ss}}}

\newcommand{\Rplus}{\R_{\geq 0}}

\newcommand{\inrate}{u_0}

\newcommand{\agents}{\mathcal{A}}
\newcommand{\strategies}{\mathcal{S}}
\newcommand{\stratprof}{\varphi}
\newcommand{\packetprof}{\sigma}

\renewcommand{\P}{\mathsf{P}}
\newcommand{\Path}{\mathsf{P}}
\newcommand{\Wait}{\mathsf{w}}

\newcommand{\lx}{x}
\newcommand{\dx}{x^d}
\newcommand{\dq}{q^d}
\renewcommand{\lq}{q}

\renewcommand{\r}{r}
\newcommand{\gap}{r_{\textrm{gap}}} 
\newcommand{\Hloc}{\mathcal{H}} 

\newcommand{\entry}[1]{\vartheta({#1})}

\newcommand{\lab}{l^{\lozenge}}

\newcommand{\lambdamax}{\lambda_\mathrm{max}} 

\newcommand{\Etildeind}{\Eloc_{\mathrm{ind}}} 
\newcommand{\Einfind}{\Einf_{\mathrm{ind}}} 

\newcommand{\lstar}{\l^*}
\newcommand{\lstarvar}{\l^\star}
\newcommand{\lstarind}{\l^{*}_{\mathrm{ind}}} 
\newcommand{\lstarss}{\l^{\circ}} 
\newcommand{\lss}{\l^{\bullet}}

\newcommand{\thetastart}{\theta_{\mathrm{start}}} 
\newcommand{\thetaend}{\theta_{\mathrm{end}}} 
\newcommand{\thetass}{\theta_{\mathrm{ss}}} 

\newcommand{\rsub}{r'}

\newcommand{\rnew}{r_{8}}
\newcommand{\reight}{r_9}

\newcommand{\Einfbar}{\bar E^\infty}
\newcommand{\aepsilon}{a}

\newcommand{\xiend}{\xi_{\mathrm{end}}}

\newcommand{\lambdadiff}{\lambda_{\mathrm{diff}}}

\newcommand{\jump}{K} 

\newcommand{\TrajLipDelta}{\hat\kappa} 

\ifdefined\anonymous

\title{Convergence of Approximate and Packet Routing Equilibria to Nash Flows Over Time}
\author{Anonymous}

\else
\title{Convergence of Approximate and Packet Routing Equilibria to Nash Flows Over Time
\thanks{An extended abstract of this work appeared at FOCS 2023, \cite{OSV23}.\\ N.O.\ is partially supported by NWO Vidi grant 016.Vidi.189.087. L.V.\ was partially supported by the Center for Mathematical Modeling at the University of Chile, Grant ANID FB210005}}

\author[1]{Neil Olver}
\author[2]{Leon Sering}
\author[3]{Laura Vargas Koch}

\affil[1]{Department of Mathematics, London School of Economics and Political Science}
\affil[2]{Department of Mathematics, ETH Z\"urich}
\affil[3]{Research Institute of Discrete Mathematics, University of Bonn}
\fi

\date{ }
\begin{document}

\maketitle

\thispagestyle{empty}

\begin{abstract}
    We consider a dynamic model of traffic that has received a lot of attention in the past few years. Infinitesimally small agents aim to travel from a source to a destination as quickly as possible. Flow patterns vary over time, and congestion effects are modeled via queues, which form based on the deterministic queueing model whenever the inflow into a link exceeds its capacity.

Are equilibria in this model meaningful as a prediction of traffic behavior? For this to be the case, a certain notion of stability under ongoing perturbations is needed. Real traffic consists of discrete, atomic ``packets'', rather than being a continuous flow of non-atomic agents. Users may not choose an absolutely quickest route available, if there are multiple routes with very similar travel times. We would hope that in both these situations --- a discrete packet model, with packet size going to 0, and $\epsilon$-equilibria, with $\epsilon$ going to 0 --- equilibria converge to dynamic equilibria in the flow over time model. No such convergence results were known.

We show that such a convergence result does hold in single-commodity instances for both of these settings, in a unified way. More precisely, we introduce a notion of ``strict'' $\epsilon$-equilibria, and show that these must converge to the exact dynamic equilibrium in the limit as $\epsilon \to 0$. We then show that results for the two settings mentioned can be deduced from this with only moderate further technical effort.
\end{abstract}

\pagenumbering{arabic}

\newpage

\section{Introduction}\label{sec:intro}

Telecommunications networks and transportation networks are two settings where the natural description involves tracking users or packets as they traverse the network.
These users arrive at different nodes in the network at different moments in time.
In some situations, this temporal aspect can be to some extent ignored, and modeled through \emph{static} models.
This is reasonable if we anticipate that over the timescale being modeled, the solution of interest can reasonably be approximated by a temporally repeated flow. 

We will be interested in the game-theoretic perspective, considering that the network traffic consists of \emph{self-interested} users, each aiming to optimize their own objective (generally, travel time) subject to the environment induced by the other users.
The interaction between agents is mediated through some form of congestion in the network.
With static flow models, this leads to the very well-studied area of network congestion games~\cite{RoughgardenBook}.

Static models do not always suffice, however.
For example, in telecommunications networks with demands changing over short time scales; or modeling morning or evening rush hour traffic.
Here, there is no plausible static approximation, and the variation on congestion over time must be considered.

In both telecommunications networks and transportation networks, many different dynamic models have been studied.
Our focus in this work will be on two related models, one continuous and the other discrete (in some sense).

\paragraph{The deterministic queueing model.}
This model goes back to Vickrey~\cite{vickrey1969congestion}, who studied this model for a single link under departure time choice.
    As well as the deterministic queueing model, it goes variously by the names of the \emph{fluid queueing model} and the \emph{Vickrey bottleneck model}.
In this model, each link has a \emph{capacity} and a \emph{transit time}.
If the inflow rate into the link always remains below its capacity, then the time taken to traverse the link is constant, as given by the transit time.
However, if the inflow rate exceeds the link capacity for some period, a queue grows on the entrance of the link.
The delay experienced by a user is then equal to the transit time, plus whatever time is spent waiting in the queue.
As long as there is a queue present it will empty at rate given by the link capacity; depending on whether the inflow rate is smaller or larger than the capacity, this queue will decrease or increase in size. 
Note that this model is nonatomic, in the sense that individual users are infinitesimally small.

There are many works investigating properties of equilibria in this model~\cite{bhaskar2015stackelberg,cominetti2015existence,cominetti2021long,correa2019price,Kaiser2022Computation,koch2012phd,koch2011nash,OSV21,sering2018multiterminal} and in generalized models \cite{
israel2020impact,pham2020dynamic,sering2019spillback,sering2020diss}.
We will discuss some of these later in \Cref{sec:exact}.

\paragraph{Packet-routing models.}
We will use ``packet-routing'' or ``packet-based'' to refer to models of a similar form to the Vickrey bottleneck model, but with atomic, unsplittable agents (or packets).
As one simple example of such a model,
suppose all links in the network have an integer capacity and an integer travel time.
Packets have unit size, and the capacity of a link represents the number of packets that can simultaneously be processed by the link in unit time (or equivalently, in a single time step; the model can be considered to be in discrete time).
If more packets than the capacity of a link need to be processed in a time step, the excess packets wait in a first-in-first-out (FIFO) buffer.
%
%
%
%
Various models of this type, varying in the details, have been considered, both in the telecommunications and transportation context, e.g., \cite{caoatomic,hoefer2011competitive,Ismaili2017RoutingGamesOverTimeFIFO,KulkarniM15,LeightonMR94,LeightonMR99,scarsini2018dynamic,tauer2021fifo,werth2014atomic}.
The traffic simulator MATSim~\cite{horni2016matsim} uses an atomic model; each ``packet'' represents a single vehicle. 
%

\medskip

So broadly speaking, 
we have described nonatomic and atomic variants of the same underlying dynamic model. 
The nonatomic nature of the deterministic queueing model is motivated primarily by 
better mathematical properties rather than as a reflection of reality.
Individual vehicles are of course not really infinitesimal; though it seems reasonable to represent them as such, as long as traffic volumes are large enough that each individual road user alone is insignificant. 

But while it seems \emph{reasonable} to expect the (nonatomic) deterministic queueing model to be a good approximation to a corresponding (atomic) packet-routing model, is this actually true? 
Can this approximation be justified?
Formally, consider the following question.
Fix a network, including arc transit times and capacities, and the inflow rate at the source.
Now consider a sequence $\beta_1, \beta_2, \ldots > 0$, with $\beta_i \to 0$ as $i \to \infty$.
For each $\beta_i$, consider an instantiation of a specific packet-routing model with packets having size $\beta_i$.
We maintain the inflow rate, measured as the product of packet size with the number of packets entering the network per unit time.
An equilibrium solution can be determined for this packet model, and if we fix any link $e$ in the network, we can observe how the length of the queue on this arc behaves in this equilibrium (a time-varying quantity).
As $i \to \infty$, does this function converge (say in the uniform norm) to the corresponding queue delay function for $e$ in the equilibrium of the deterministic queueing model?
If this is \emph{not} true, then one has to seriously question the relevance of the deterministic queueing model.

Positive experimental evidence for convergence was found in \cite{ZiemkeEtAl2020FlowsOverTimeAsLimitOfMATSim}.
In \cite{sering2021convergence}, convergence was shown for a \emph{fixed} choice of paths for all packets (in an appropriate sense). 
This already involves some significant technicalities, but their result does not say anything about the relationship between \emph{equilibria} in the two settings.
A key difficulty is that we do not know a priori that the paths chosen in the equilibrium of the packet model will resemble those chosen in the equilibria of the deterministic queueing model.


There are a number of other distinct but similar questions one can ask, all concerning the stability of the deterministic queueing model and its equilibria. 
In exact equilibria, users choose exactly quickest paths to the sink.
That is, given the strategies chosen by the other users, they choose a path that in hindsight yields the earliest possible arrival time.
This is quite a strong property; note that users are taking into account queues that they will see ``in the future'', not the queues as they are on entry into the network.
One motivation for this is that we view, for example, morning rush hour traffic as a repeated game, with the expectation that behavior converges to a Nash equilibrium.\footnote
{
Other equilibria notions distinct from Nash equilibria have been considered in the literature, in which agents make decisions without full information of the overall traffic situation.
We refer to Graf, Harks and Sering~\cite{graf2020dynamic} and references therein for a discussion of instantaneous dynamic equilibria (see also \cite{graf2023finite,graf2023price}), where agents make decisions as they traverse the network based on current queues; 
and to Graf, Harks, Kollias and Merkl~\cite{Graf_Harks_Kollias_Markl_2022} for a very interesting approach to a much more general information framework where users use predictions of future congestion patterns.
}
Still, it seems implausible to expect that this process always achieves \emph{exactly} a Nash equilibrium.
It seems much more plausible to hope that we obtain an approximate equilibrium; no user is taking a path that is very far from their quickest option, but might be choosing a route that is close to quickest, but not quite.
So it is natural to consider $\epsilon$-equilibria in this model (or in the packet model), and ask if their behavior is similar to that of exact equilibria.
Again as a precise question: do $\epsilon$-equilibria of the deterministic queueing model converge to the exact equilibrium, in the same sense as above?

Other natural types of ``perturbations'' can be considered.
For instance: arc travel times and capacities might vary slightly over time, in some predictable or unpredictable way, or the demand might vary slightly over time instead of being precisely constant.
The question in each case is the same: for sufficiently small perturbations of whatever form, can we say that equilibria in the perturbed system are close to equilibria in the original unperturbed system?

\paragraph{Our results.}
We give a positive answer to the following two convergence questions in the single-commodity setting.
We show that equilibria in a particular packet model converge to that of the deterministic queueing model, as the size of the packets goes to zero; and that $\epsilon$-equilibria converge to the exact equilibrium as $\epsilon \to 0$.
Moreover, we do this in a unified way.
We will prove a single main convergence result, and then show that both of these specific results follow.

Our convergence result is based around the notion of a \emph{strict $\delta$-equilibrium}, which is a fairly natural strengthening of an $\epsilon$-equilibrium.
It asks that for \emph{every} node in the network that an agent uses in their path, not just the sink, the agent's departure time from that node is at most $\delta$ later than its earliest possible arrival time (considering all possible routes to the node).
This is stronger than asking for an $\epsilon$-equilibrium with $\epsilon=\delta$, which only requires this at the sink.
It is not the case that every $\epsilon$-equilibrium is a strict $\delta$-equilibrium with $\delta=\epsilon$, but we are able to show that every $\epsilon$-equilibrium is a strict $\delta$-equilibrium with $\delta=O(\epsilon)$ (here, the big-O notation hides network-dependent constants).
So convergence results for strict $\delta$-equilibria hold for $\epsilon$-equilibria as well. 

To obtain results for packet routing, we ``embed'' an equilibrium of the packet-routing model into our continuous framework, viewing each packet as consisting of a continuum of particles.

While we focus on these two specific implications (and also use a specific packet-routing model), our convergence result can certainly be used to derive other stability results.
For example, convergence results for other packet-routing variants, and for perturbed transit times and inflow rates.
We will not consider this further in the current paper however, in order to focus on our central results.

One challenge in proving such a result is that the perturbations we are considering are ongoing throughout the evolution of the equilibrium.
A much weaker notion of stability would be that if we slightly perturb the equilibrium at a single moment in time, or some bounded number of moments in time, by (say) perturbing some queue lengths or transit times by a small amount, that the equilibrium in the once-perturbed instance stays close to the unperturbed equilibrium.
We demonstrated this quite recently for the deterministic queueing model~\cite{OSV21}. 
Our earlier result can be seen as a precursor to this one, and we will rely on it in a number of places in our proof.
However, their result is not strong enough to handle the convergence results we are interested in here. 
In some sense, we need to show not only that the perturbations that occur in our perturbed equilibria at a particular moment do not lead to vastly different behaviors, but rather that there is a tendency to ``revert to the mean'': if some queue gets a little longer than it should in the perturbed situation, future perturbations might push this back towards the unperturbed value, 
but do not push it even further away.

A smaller technical issue we have to address involves the foundations of the definition of the model.
Especially for the implications to convergence of packet-routing models, it is necessary for us to allow waiting in our nonatomic model;
an agent (that is, an infinitesimal flow particle) is allowed to wait at a node before entering the next arc of its path to the sink.
This expansion of the strategy space, from a finite set of paths to a finite \emph{dimensional} space of paths along with waiting times at nodes, requires some technical care, and we make some efforts to handle this in a precise and clean way.
As one example of a complication that arises, it is now possible for a positive measure of agents to enter an arc at precisely the same moment in time. 
Previous works in this area had no particular need to consider waiting, and did not face this issue.

\medskip

A different application of our results is in the other direction, to port results on the deterministic queueing model to packet models.
This allows us to profit from the cleaner and more analytically tractable setup of the continuous model. 
Here we briefly discuss two implications for packet models; we expect there will be more.

Suppose we are considering an instance of the packet model, in which packets enter the network at $s$ at a constant rate, and wish to reach a sink $t$. 
Suppose further that the number of packets entering the network per unit of time, multiplied by their size, is not larger than the minimum capacity of an $s$-$t$-cut.
Is it then true that queues in the network remain bounded for all time?
For the deterministic queueing model analog of the instance, this is known to be true~\cite{cominetti2021long}. 
The proof uses a very delicate potential function, obtained from the dual of a linear program that describes so called ``steady-state'' conditions. 
It is not clear how this argument could be directly ported to discrete packet models.
But our convergence result implies that indeed queues do remain bounded in the packet model\,---\,at least, as long as $\delta$ is sufficiently small. 
We expect that with some further technical effort, our convergence result can be strengthened so that this restriction can be bypassed.

A second question that can be attacked with this approach is that of the price of anarchy.
Here, one needs to specify the objective to compare with some care.
It is known that the ratio between the average journey time of agents in an equilibrium, compared to the global optimum, can be unbounded even on very simple examples~\cite{koch2012phd}.
However, if one considers the average arrival time objective (equivalently, viewing packets as all being at the source at time $0$), this becomes an interesting question.
It remains open in the deterministic queueing model, but it is conjecture that the price of anarchy for this measure is precisely $\frac{e}{e-1}$, and it is known that this holds if another natural conjecture is true~\cite{correa2019price}.
If this conjecture is demonstrated, it will immediately imply (through \cite{correa2019price} and our result) that the price of anarchy is bounded in the packet model, again as long as $\delta$ is sufficiently small.

\medskip

The question of whether an equilibrium is ``stable'' under some form of perturbation is a rather natural one, also in other non-traffic settings.
Aswathi, Balcan, Blum, Sheffet and Vempala~\cite{ApproxStable} and Balcan and Braverman~\cite{PerturbationStable} (see also \cite{Liptonetal}) explicitly introduce and investigate a related notion in the context of bimatrix games.
In that context, they say that a Nash equilibrium is \emph{$(\delta, \epsilon)$-perturbation stable} if whenever all payoffs in the bimatrix game are adjusted by at most $\delta$, any equilibrium in the resulting game is within distance $\epsilon$ (in variation distance) of an equilibrium of the unperturbed game. 
These papers study various properties (especially computational properties) of perturbation stable games.

\medskip

Dynamic traffic modeling is a huge, multidisciplinary area, and we do not attempt to do it justice in this brief survey.
In particular, our discussions have focused on the work done by the algorithmic game theory community.
We refer the reader to the survey by Friesz and Han~\cite{Friesz_Han_2019} for a different perspective on the topic, considering a more general class of link dynamics through the lens of differential variational inequalities.
%

\section{Model and preliminaries}\label{sec:model}

An instance is described by a directed network $G=(V,E)$, with arc capacities $\nu_e > 0$ and free-flow travel times $\tau_e > 0$ for all arcs $e \in E$\footnote{Excluding arcs with $\tau_e = 0$ is convenient for technical reasons; it should be possible to extend to at least the setting where there are no directed cycles of $0$-length arcs, but we will not discuss this here.}.
In addition, there is a specified source node $s \in V$ and sink node $t \in V$, and a constant \emph{network inflow rate} $u_0$.
We may assume that every node in $G$ is both reachable from $s$, and can reach $t$.

We use the notation $\delta^-(v)$ and $\delta^+(v)$ to denote the set of incoming and outgoing arcs at $v$, respectively, and similarly $\delta^-(S)$ and $\delta^+(S)$ for arcs entering or leaving a set $S \subseteq V$.

Whenever not specified, we will use $\|\cdot\|$ to refer to the infinity norm, which will be our main measure of distance. 
Given a point $x \in \mathbb{R}^m$ and a set $S \subseteq \mathbb{R}^m$, we use $\dist(x,S)$ to denote the distance (with respect to the infinity norm) between $x$ and $S$, that is: $\dist(x,S) := \inf_{y \in S} \|y-x\|$.
Similarly, given two sets $S,T \subseteq \mathbb{R}^m$, $\dist(S,T) := \inf_{x \in S} \dist(x,T) = \inf_{x \in S, y \in T} \|y-x\|$.
We will use $B_r(x)$ to denote the ball of radius $r$ around $x \in \mathbb{R}^m$ and $B_r(S) = \set{x \in \R^m | \dist(x,S) \leq r}$ for any $S \subseteq \R^m$, both with respect to the infinity norm.

\subsection{Flows over time with waiting}
In the literature (e.g.,~\cite{cominetti2015existence}), flows over time are typically denoted by a family of functions $(f_e^+, f_e^-)_{e \in E}$, where $f_e^+(\xi)$ denotes the inflow rate into arc $e$ at time $\xi$ and $f_e^-(\xi)$ the flow rate out of arc $e$ at time $\xi$.
As we want to allow particles to wait at nodes, this choice becomes less convenient, as it is possible that particles wait at a node in such a way that an atom of particles enters an arc at the same moment in time. 
In this case the inflow rate would be infinite.
Instead, 
we define flows in terms of \emph{cumulative} flow functions which are essentially the integrals of $f_e^+$ and $f_e^-$.

A \emph{flow over time with waiting} consists of a pair $(F^+, F^-)$, where $F^+$ is a vector of functions $F^+_e: \Rplus \to \Rplus$ for arcs $e \in E$, and similarly for $F^-$. 
For each arc $e$ and $\xi \in \Rplus$, $F_e^+(\xi)$ denotes the total amount of flow that has entered arc $e$ up to time $\xi$ and $F_e^-(\xi)$ denotes the total flow amount that has left arc $e$ up to time $\xi$.
Each $F^+_e$ and $F^-_e$ should be nondecreasing and right-continuous function.
These functions must satisfy the following two conditions.

\paragraph{Relaxed flow conservation.}
For all times $\xi \in \Rplus$ it must hold that
\begin{equation} \label{eq:weak_flow_conservation}
\sum_{e\in \delta^-(v)} F^-_e(\xi) - \sum_{e\in \delta^+(v)} F^+_e(\xi) \geq 
\begin{cases}
0 & \text{ for all } v \in V \setminus\set{s,t},\\
- \inrate\xi & \text{ for } v = s.
\end{cases}
\end{equation}
Note that we require the flow to enter the network at $s$ with constant inflow rate of $\inrate$. 

\paragraph{Queues operate at capacity.}
We assume that arcs always operate at capacity (waiting is allowed at nodes, but there is no waiting on arcs in our model). 
Let $z_e(\xi)$ be the \emph{queue volume} on $e$ at time $\xi$; that is, the total measure of particles in the queue at time $\xi$.
We have 
\[ z_e(\xi) := F_e^+(\xi) - F_e^-(\xi + \tau_e);
\]
particles that enter by time $\xi$, but have not left the queue by time $\xi$ (and hence have not left the tail of the arc by time $\xi + \tau_e$) contribute to the queue volume.

For all $e \in E$ and all times $\xi \in \Rplus$ we require that
\begin{equation}
\label{eq:queues_operate_at_capacity}
z_e(\xi) = \max_{0 \leq \psi \leq \xi}\Bigl( F_e^+(\xi) - F_e^+(\psi) - \nu_e(\xi - \psi) \Bigr). 
\end{equation}
The interpretation of this is that for any $\psi \leq \xi$, $z_e(\xi)$ is at least the mass of particles entering in the interval $[\psi, \xi]$, minus the upper bound $\nu_e(\xi - \psi)$ on the mass of particles that can leave the queue in this time.
Further, if $\psi$ is chosen so that $z_e(\psi) = 0$ but $z_e(\xi') > 0$ for all $\xi' \in (\psi, \xi)$, then we do not merely have a lower bound on $z_e(\xi)$, but must have equality, since the queue must operate at capacity on the interval $(\psi, \xi)$.



\subsection{Agent perspective}
A flow over time with waiting does not identify a path or flow corresponding to a given particle. 
For exact dynamic equilibria, this is not a concern; a flow over time that corresponds to a dynamic equilibrium provides sufficient information to reconstruct the flow attributable to departures from the source at any moment in time.
This is no longer the case for our setting however, and we need additional direct information about particle behavior.

We will denote our set of agents (equivalently, particles) by $\agents := \Rplus \times [0,1]$.
We let $\mu$ denote the Lebesgue measure on $\agents$.
For each $a \in \agents$, we use $\entry{a}$ to denote the first coordinate of $a$ divided by $u_0$, which we will interpret as the \emph{entry time} of agent $a$ into the system, i.e., the time it arrives at the source.
(Put differently, the first coordinate of $a \in \agents$ indicates the measure of particles that arrive at the source before $a$.)
Previous works on Nash flows over time generally took the set of particles to be indexed by $\Rplus$, identifying an agent with its entry time.
The strategy of a flow particle was then described by a unit flow. 
This approach turns out to be inconvenient for our more general setting, however. 

A \emph{strategy} for an agent consists of a pair $(P,w)$, where $P$ is an $s$-$t$-path, and $w \in \Rplus^{V(P)}$ denotes the amount of time that the agent will wait at each node in the path.
Let $\strategies$ denote the set of all possible strategies.
We view $\strategies$ as a measurable space, where a set $Q \subseteq \strategies$ is measurable if $\{ w \in \Rplus^{V(P)} : (P,w) \in Q \}$ is Lebesque measurable for every $s$-$t$-path $P$.
A \emph{strategy profile} $\stratprof$ is a measurable map from $\agents$ to $\strategies$.
We use $\P^\stratprof(a)$ to denote the first component of $\stratprof(a)$, i.e., the path that agent $a$ chooses.
For each $v \in V$, we define $\Wait^\stratprof_v$ to be the partial function that defines $\Wait^\stratprof_v(a)$ to be the time agent $a$ waits at $v$, if $v \in \P(a)$.
We may write $\Wait^\stratprof(a)$ for the vector $(\Wait^\stratprof_v(a))_{v \in V(\Path^\stratprof(a))}$.
We will typically omit the explicit dependence on $\stratprof$ in our notation whenever it is unambiguous.

The measurability condition on $\stratprof$ implies that for any $\theta_1 \leq \theta_2$, $s$-$t$-path $P$, and any Lebesgue measurable set $R \subseteq \Rplus^{V(P)}$, 
\[ \{ a \in \agents: \entry{a} \in [\theta_1, \theta_2], \P(a) = P \text{ and } \Wait(a) \in R\} \]
is a measurable set.

Note that $\mu(\{ a \in \agents: \theta_1 \leq \entry{a} \leq \theta_2\}) = u_0(\theta_2 - \theta_1)$ for all $\theta_1 \leq \theta_2$, given the network inflow rate of $u_0$.
In particular, the set of particles entering the network at some time $\theta$ is always a null set.

\medskip

An \emph{outcome} of the game for a given strategy profile $\stratprof$ 
specifies, for each particle $a$, their precise departure time from each node $v$ on their path $\P(a)$.
This must correspond to a flow over time with waiting as described above.
We now make this precise.

We specify an outcome by a flow over time with waiting $(F^+, F^-)$, and partial functions $d_v: \agents \to \Rplus$ for each $v \in V$.
The value $d_v(a)$ is defined only when $v \in V(\P(a))$, and in that case, it describes the time at which agent $a$ departs $v$ and enters the arc $e=vw \in \P(a)$ that follows, or the time that the agent departs the network if $v=t$.
We call each $d_v$ a \emph{departure time function}. 

In order for $(F^+, F^-, d)$ to represent a valid outcome of a given strategy profile $\stratprof$, we require the following to hold.
For each arc $e$, let $z_e$ be the queue volume for $e$ associated with $(F^+, F^-)$.

\begin{itemize}
    \item Departure times must be consistent with queue delays and node waiting times. 
        Consider any agent $a$, and arc $e=vw \in \P(a)$. 
    We must have that 
    \[ d_w(a) = d_v(a) + \dq_e(a) + \tau_e + \Wait_w(a), \]
    where $\dq_e(a)$ is the amount of time that $a$ waits on the queue on arc $e$. 

    The value of $\dq_e(a)$ is essentially determined by the queue volume at the time $d_v(a)$ that agent $a$ enters $e$, with the additional complication that if an atom of particles enters $e$ at this same moment, a tiebreaking rule is required.
    We tiebreak according to entry time into the network. 
    (No tiebreaking is required between agents with the same entry time, since this is always a null set.)
    So we have
    \begin{equation}\label{eq:queue-waiting}
        \dq_e(a) := \tfrac{1}{\nu_e}\bigl(z_e(d_v(a)) - \mu(\{ a' \in \agents: e \in \Path(a'), d_v(a') = d_v(a) \text{ and } \entry{a'} > \entry{a}\})\bigr) . 
    \end{equation}

\item The cumulative flow $F^+_e(\xi)$ entering an arc $e=vw$ by some time $\xi$ matches with $\stratprof$ and $d_v$.
    That is, 
    \[ F_e^+(\xi) = \mu(\{ a \in \agents : e \in \P(a) \text{ and } d_v(a) \leq \xi\}).
    \]
\end{itemize}

\subsection{Network loading}
It is not immediately obvious how to construct the outcome $(F^+, F^-, d)$, nor even that they exist or are unique. 
The demonstration of this is via a \emph{network loading} procedure.
This is fairly standard, and there are no major conceptual issues, but previous discussions of network loading that we are aware of do not allow for waiting, and this does introduce some minor technical complications. We defer the proof to the appendix.

\begin{restatable}{thm}{loading}
\label{thm:loading}
    Given any strategy profile $\stratprof$, there is a unique associated outcome $(F^+, F^-, d)$.
\end{restatable}

\subsection{A form of approximate dynamic equilibria}\label{sec:delta-trajectories}

We now recall the notion of \emph{earliest arrival labels}, ubiquitous in the study of Nash flows over time (see \cite{cominetti2015existence,cominetti2021long,koch2011nash}
among others).
Let $\stratprof$ be a strategy profile, with outcome $(F^+, F^-, d)$, and let $(z_e)_{e \in E}$ be the queue volume functions associated with this.
Then for any $v \in V$, the earliest arrival label $\l_v: \Rplus \to \Rplus$ maps an entry time $\theta$ to an earliest possible time a hypothetical particle departing at time $\theta$ could arrive at $v$, taking into account queueing delays induced by other agents using the current strategy profile.
They can be defined via the Bellman equations
\begin{equation}\label{eq:bellman}
\l_w(\theta) =\begin{cases} \theta & \text{ if } w=s\\
    \min_{e=vw} \l_v(\theta) + \tau_e + z_e(\l_v(\theta))/\nu_e & \text{ otherwise.}
\end{cases}
\end{equation}
Note that $z_e(\l_v(\theta))/\nu_e$ is the queue waiting time a hypothetical particle departing the source at time $\theta$ and arriving at the earliest possible time $\l_v(\theta)$ experiences on edge $e=vw$.
There is no issue to worry about in terms of tiebreaking, since all particles with $d_v(a) = \l_v(\theta)$ will have entry time at most $\theta$, and so do delay our hypothetical particle.

\paragraph{``Exact'' dynamic equilibria.}
A dynamic equilibrium has a simple definition in our notation.
It is that $d_v(a) = \l_v(\entry{a})$ for all $a \in \agents$ and $v \in \P(a)$.
That is, each agent arrives and departs at each node on its path at an earliest possible time (in particular, the agent arrives at the sink at the earliest possible time) taking into account queueing delays.

Given the vector $\l$ of earliest arrival labels of a dynamic equilibrium, 
we will follow \cite{OSV21} in calling $\l$ an \emph{equilibrium trajectory}.
We will discuss properties of dynamic equilibria and equilibrium trajectories in more detail in \Cref{sec:exact}.

\paragraph{$\epsilon$-equilibria.}
We can easily interpret the general notion of an $\epsilon$-approximate Nash equilibrium (more briefly, an $\epsilon$-equilibrium) in our model.
Every agent should have a travel time that is at most $\epsilon$ larger than the best travel time they could achieve, taking into account the actions of all other agents. 
In other words, a strategy profile is an $\epsilon$-equilibrium for some $\epsilon > 0$ if the outcome satisfies
\begin{equation}\label{eq:epsilon-eq}
    d_t(a) \leq \l_t(\entry{a}) + \epsilon \quad \text{ for all $a \in \agents$.}
\end{equation}

\paragraph{Strict $\delta$-equilibria.}
If we consider some arbitrary node $v$ in an $\epsilon$-equilibrium, it need not be the case that every agent $a$ that uses $v$ in their path arrives at $v$ within $\epsilon$ of the earliest possible arrival time.
The reason is that the agent may be able to ``catch up'' by the time it reaches the sink.
For example, if an arc entering the sink has large capacity, but at some point in time has a large queue, 
then agents could join the back of this queue over a larger interval of time, but exit the queue over a shorter interval.
%
%

It will be useful for our purposes to consider the stronger notion where this property does hold.
Define a \emph{strict $\delta$-equilibrium} as a strategy profile where the outcome satisfies
\begin{equation}\label{eq:delta-eq}
 d_v(a) \leq \l_v(\entry{a}) + \delta \quad \text{ for all $a \in \agents$ and $v \in \P(a)$.} 
\end{equation}

Given a strict $\delta$-equilibrium, the corresponding earliest arrival labels $\l$ will be of particular importance for us (as they were in the case of exact dynamic equilibria).
If $\l$ arises from a strict $\delta$-equilibrium, we will say simply that $\l$ is a \emph{$\delta$-trajectory}.
%
%


\subsection{Properties of exact equilibria}\label{sec:exact}
We now briefly summarize some useful facts about the structure of (exact) equilibria.
For more details, we refer to \cite{cominetti2015existence} and \cite{koch2011nash} on thin flows and the piecewise-linear structure; to \cite{cominetti2021long} and \cite{OSV21} for long-term behavior; and to \cite{OSV21} for the vector-field view and uniqueness and continuity of equilibria.

In most previous discussions of dynamic equilibria in networks of Vickrey bottlenecks, there is no strategy profile in the sense we have defined it for our model, where each particle chooses a single path.
Rather, an equilibrium is described by a flow over time $(F^+, F^-)$ (without waiting), which induces the earliest arrival labels $\l(\theta)$ and associated queue volumes $z_e$.
Since there is no waiting, $z_e$ is continuous for each $e$.
Let $\lq_e(\theta) = z_e(\l_v(\theta))/\nu_e$ for each $e=vw \in E$.
An arc $e=vw$ is called \emph{active} at entry time $\theta$ if $\l_w(\theta) = \l_v(\theta) + \tau_e + \lq_e(\theta)$. 
This means that a particle departing the source at time $\theta$ has a shortest path to $w$ that uses arc $e$, and that $e$ defines $\l_w(\theta)$ in the Bellman equations~\eqref{eq:bellman}.
Then one definition of a dynamic equilibrium is that $(F^+_e)'(\l_v(\theta)) = 0$ whenever $e=vw$ is not active at entry time $\theta$.
This matches our earlier definition in \Cref{sec:delta-trajectories} for our model where waiting is allowed: agents must arrive at the earliest possible time at each node on their path.
It turns out that if one defines $x_e(\theta) := F_e^+(\ell_v(\theta))$ for all $e=vw$ and $\theta$, then for a dynamic equilibrium, $x(\theta)$ is an $s$-$t$-flow of value $u_0\theta$, for each $\theta$.
From our perspective, $x_e(\theta)$ can be viewed as the measure of agents with entry time at most $\theta$ which choose arc $e$ in their strategy.

Conveniently, as we describe below, $\l$ alone, without reference to the defining flow over time, suffices to describe an exact dynamic equilibrium (which is not the case for approximate equilibrium concepts).
Consider some $e=vw$.
If $\l_w(\theta) < \l_v(\theta) + \tau_e$, then even without a queue, $e$ is not active at entry time $\theta$.
Further, it can be argued that in a dynamic equilibrium, $\lq_e(\theta) = \max\{\l_w(\theta) - \l_v(\theta) - \tau_e, 0\}$. 
So information about whether arc $e=vw$ is active or not, whether it has a queue or not (from the perspective of a particle entering the network at time $\theta$), and the length of that queue, is completely determined by $\l(\theta)$.

\paragraph{Active and resetting arcs.}
For any $\linit \in \R^{V}$, let 
\begin{equation}
    \begin{aligned}
        E'_{\linit} &:= \{ e =vw \in E : \linit_w \geq \linit_v + \tau_e\}, \quad \text{and}\\
        E^*_{\linit} &:= \{ e =vw \in E : \linit_w > \linit_v + \tau_e\}.
    \end{aligned}
\end{equation}
So $e$ is active at entrance time $\theta$ if $e \in E'_{\l(\theta)}$, and has a queue if $e \in E^*_{\l(\theta)}$. 
We also call the arcs with a queue \emph{resetting} arcs.

\paragraph{Thin flows.}
It has been shown~\cite{cominetti2015existence,koch2011nash} that a flow over time is in equilibrium if and only if the resulting pair $(x,\l)$ satisfies the following \emph{thin flow conditions} for almost every $\theta$:
setting $x' = x'(\theta)$, $\l' = \l'(\theta)$, $E' = E'_{\l(\theta)}$ and $E^* = E^*_{\l(\theta)}$, 
%

\begin{equation} \label{eq:thin}
\begin{alignedat}{2}
    x' &\text{ is a static $s$-$t$ flow of value $u_0$,}\\
\l'_{s} &= 1, && \\ 
\l'_w &= \min_{e = vw \in E'} \rho_e(\l'_v, x'_e) \quad && \text{ for all } w \in V \setminus \set{s},
\\ 
\l'_w &= \rho_e(\l'_v, x'_e) && \text{ for all } e = vw \in E' \text{ with } x'_e > 0, 
\end{alignedat}
\end{equation}
\[\text{ where } \qquad \rho_e(\l'_v, x'_e) \coloneqq \begin{cases}
\frac{x'_e}{\nu_e} & \text{ if } e = vw \in E^*,\\
\max\bigl\{\l'_v, \frac{x'_e}{\nu_e}\bigr\} & \text{ if } e = vw \in E'\backslash E^*.  
\end{cases}\]
Note that the conditions are fully determined by the pair $(E', E^*)$, with $E^* \subseteq E'$. As long as \begin{inparaenum}[(i)] \item each node $v$ is reachable from $s$ in $E'$, \item each arc $e \in E^*$ lies on an $s$-$t$-path in $E'$, and \item
no arc of $E^*$ lies on a directed cycle in $(V,E')$, \end{inparaenum} these equations always have a solution, and $\l'$ is uniquely determined~\cite{cominetti2015existence,koch2012phd}.
We will sometimes call this unique $\l'$ (leaving out $x'$) the \emph{thin flow direction}.


We call a pair $(E', E^*)$ satisfying (i)-(iii) above a \emph{valid configuration}.

Furthermore, we call a vector $\lab \in \R^V$ \emph{valid} if $(E'_{\lab}, E^*_{\lab})$ is a valid configuration. 
We will use $\Omega \subseteq \R^V$ to denote the set of valid labels. 

Define $X: \Omega \to \R^V$ be the vector field for which $X(\linit)$ is the unique solution to the thin flow equations for $(E'_{\linit}, E^*_{\linit})$, for all $\linit \in \Omega$.
Then put differently $\l'(\theta) = X(\l(\theta))$ for almost every $\theta$.
Since $X(\linit)$ depends only on $E'_{\linit}$ and $E^*_{\linit}$, it is piecewise constant, and indeed with a very specific structure.
Each arc $e=vw$ divides $\Omega$ into two open halfspaces separated by the hyperplane $\{ \linit \in \Omega: \linit_w - \linit_v = \tau_e \}$.
So an equilibrium trajectory $\l$ has a piecewise linear structure, with its direction only changing upon hitting a hyperplane.
Each maximal piecewise-linear segment of $\l$ is called a \emph{phase}.

We can define an equilibrium trajectory starting from any initial point $\linit \in \Omega$, not necessarily an empty network. 
This can be interpreted, with some care, as starting with some initial queues present; if $\l_w(0) - \l_v(0) - \tau_e > 0$, this value represents a queue delay that an agent starting at time 0 and traversing $e$ via a shortest path would experience on the arc, not the queue length at time $0$. 

\paragraph*{Generalized subnetworks.}
A further generalization we will need in our arguments is the notion of a \emph{generalized subnetwork}, 
 following \cite{OSV21}.
A \emph{generalized subnetwork} of our network $G$ is defined by valid configuration $(\Etilde, \Einf)$.  
Given such a pair, we can define a new vector field $X^{(\Etilde, \Einf)}(\cdot)$, by defining its value at position $\linit$ to be the solution to the thin flow equations determined by the pair $(\Etilde \cap E'_{\linit}, \Einf \cup E^*_{\linit})$ (as opposed to $(E'_{\linit}, E^*_{\linit})$). 
(The point $\linit \in \R^V$ is valid for the generalized subnetwork if $(\Etilde \cap E'_{\linit}, \Einf \cup E^*_{\linit})$ is a valid configuration.)
Only arcs in $\Etilde \setminus \Einf$ will have a corresponding hyperplane; arcs in $E \setminus \Etilde$ act always as being inactive, and arcs in $\Einf$ are viewed as always having a queue.
We can define an equilibrium trajectory in this generalized subnetwork in the same way as for the full network; a trajectory $\l$ that follows $X^{(\Etilde, \Einf)}$ almost everywhere.
%

\paragraph{Long-term behavior.}
Given an equilibrium trajectory $\l$ in some generalized subnetwork $(\Etilde, \Einf)$, we say that $\l$ has reached \emph{steady state} by time $T$ if $(\l(\theta): \theta \geq T)$ is within a single phase, i.e., the trajectory does not hit any hyperplane after time $T$.
This means that queues change linearly from time $T$ forward (in particular, if $\Einf = \emptyset$ and the network has sufficient capacity, then in a steady state queues will remain constant~\cite{cominetti2021long}).

With $\tilde\Omega$ being the set of valid labels for the generalized subnetwork, 
say that a label $\linit \in \tilde\Omega$ is a \emph{steady-state label} if the equilibrium trajectory starting from $\linit$ is immediately at steady state. 
Let $I$ denote the set of steady-state labels. 
It can be shown that there is a unique ``steady-state direction'' $\lambda$ so that for every $\linit \in I$, the equilibrium trajectory starting from $\linit$ is $\l(\theta) = \linit + \lambda \theta$. See \cite{OSV21} for details.

We will need the following result from \cite{OSV21}, which builds on an earlier result by \cite{cominetti2021long}, and shows that equilibrium trajectories always reach a steady state.
(We need a slightly more refined form of it than the main statement in the paper; we discuss the details of how this theorem follows from \cite{OSV21} in the appendix.)
\begin{restatable}{thm}{steadystate}\emph{({\cite{OSV21}})} 
\label{thm:steady-state}
Consider a generalized subnetwork determined by a valid configuration $(\Etilde, \Einf)$, and let $\Iloc$ be the corresponding steady-state set.
Then there exists some $\T$ such that for any valid starting label $\linit$, the equilibrium trajectory $\l$ starting from $\linit$ reaches steady state in time at most \modified{$\T \cdot d(\linit, I)$}.
\end{restatable}

\paragraph{Continuity.}
In \cite{OSV21}, it was shown that there is a \emph{unique} equilibrium trajectory for any given starting point $\l(0) \in \tilde\Omega$, and 
that the trajectory $\l$ depends \emph{continuously} on the starting point $\l(0)$.
We will make crucial use of this.
\begin{theorem}[{\cite[Theorem 3.2]{OSV21}}]\label{thm:continuity_of_Nash}
Fix a generalized subnetwork and let $\Psi: \tilde\Omega \to L^\infty([0, \infty))$ be the map that takes $\linit \in \tilde\Omega$ to the unique equilibrium trajectory $\l$ satisfying $\l(0) = \linit$. Then $\Psi$ is a continuous map, where we imbue $L^\infty([0, \infty))$ with the supremum norm.
\end{theorem}

\section{Technical overview}\label{sec:overview}

 \subsection{The main convergence result}
 We are now ready to state our main result in its precise form. All our results here apply to an arbitrary fixed instance; note that constants hidden in $O(\cdot)$ typically depend on that instance.
 
 \begin{theorem}\label{thm:main-full}
  Let $\linit$ be the labeling corresponding to the empty network, and let $\lstar$ be the equilibrium trajectory starting from $\linit$.
  Then for every $\epsilon > 0$, there is a $\delta > 0$ such that every $\delta$-trajectory $\l$ starting from $\linit$ stays within distance $\epsilon$ of $\lstar$, i.e., 
  \[\norm{\l(\theta) - \lstar(\theta)} \leq \epsilon \quad \text{ for all } \theta \in \Rplus.\]
 \end{theorem}
 
 
 \subsection{Implications for \texorpdfstring{$\epsilon$}{epsilon}-equilibria and packet models}
To show that we can use the above theorem to obtain convergence for our two applications $\epsilon$-approximate equilibria and packet routings, we prove that both these concepts can be modeled as strict $\delta$-equilibria, where $\delta$ depends on $\epsilon$ in the first case and on the packet size in the latter case.

\begin{restatable}{thm}{epsilonequiisdeltatrajectory}
\label{thm:epsilon_equi_is_delta_trajectory}
Let $\stratprof$ be an $\epsilon$-equilibrium for some $\epsilon > 0$. 
Then $\stratprof$ is a strict $\delta$-equilibrium for $\delta = O(\epsilon)$. 
\end{restatable}
(The big-O hides network-dependent constants.)
To prove this, we show that for an $\epsilon$-equilibrium, (i) the mass of particles that could in principle overtake a fixed agent at a given node is $O(\epsilon)$, and (ii) the earliest arrival labels fulfill an approximate Lipschitz-property. 
If an agent $a$ were to arrive at some node much later than the quickest path would allow, then by approximate Lipschitz-continuity, the measure of other agents that would be able to overtake $a$ would be too large.

\medskip

We will fix a packet-routing model that is similar to the one discussed in \cite{hoefer2011competitive}; see \Cref{subsec:packet_model} for a formal definition of the model. 
In this model, a packet enters the next arc of its path only once it has been fully processed by the previous arc.
Given an equilibrium in this packet model, say with packets of size $\beta$, 
we can view each packet as consisting of a measure $\beta$ of infinitesimal flow particles, each taking the same path. 
In order to maintain the temporal integrity of a packet, we exploit the flexibility of waiting at nodes in our model.
If a packet is being processed by some arc $e=vw$, we hold all particles of the packet at $w$ as long as any of the particles are still being processed by arc $e$.
Once all these particles arrive, they depart all at once onto the next arc of the path.
This results in a joint strategy choice for all particles, with waiting at nodes.

\begin{restatable}{thm}{packets} \label{thm:packets}
Suppose we are given an equilibrium of the packet model with packet size $\beta$, and consider the corresponding flow over time strategy profile $\stratprof$.
Then $\stratprof$ is a strict $\delta$-equilibrium for $\delta = O(\beta)$.
\end{restatable}
We show that $\stratprof$ is an $\epsilon$-equilibrium for some $\epsilon = O(\delta)$; the claim then follows from \Cref{thm:epsilon_equi_is_delta_trajectory}.
The intuition for why this holds is simply that the ``last'' particle in each packet takes an earliest arrival path\footnote{``Last'' rather than ``first'' because of details of the specific packet model; for other natural packet models this might hold for the first particle instead.}; and for other particles, assuming some Lipschitzness, things cannot go too badly wrong.

 \subsection{Proof overview of the main convergence result}
 It will be somewhat convenient for our purposes to invert the dependence between $\epsilon$ and $\delta$. 
 We will think of $\delta > 0$ as being given, and we must choose $\epsilon$ depending on $\delta$ so that any $\delta$-trajectory $\ell$ remains $\epsilon$-close to the equilibrium trajectory $\lstar$, and moreover, the dependence of $\epsilon$ on $\delta$ must be such that $\epsilon  \to 0$ as $\delta \to 0$.
 From this perspective, $\delta$ is some ``small'' quantity, and $\epsilon$ will be some typically much larger quantity --- but nonetheless still ``quite small'' in the sense that it goes to $0$ as $\delta$ goes to $0$.
 Our arguments will involve producing a sequence of parameters that are ``quite small'' in the same sense, eventually leading to our choice of $\epsilon$. 
 We make this precise with the following definition.
 
 \begin{definition}
 We call a function $r \colon \R_{> 0} \to \R_{\geq 0}$ a \emph{small parameter} if $r(\delta) \to 0$ as $\delta \to 0$, and its definition only depends on the network $G$ and previously defined small parameters.
 
 We will typically omit the dependence on $\delta$ if there is no ambiguity, and write simply $r$ rather than~$r(\delta)$.
 \end{definition}
 
Let $\Tglob$ denote the time an equilibrium trajectory in the network $G$ starting from the empty network requires to reach steady state.
Our main convergence theorem comes as a consequence of the following technical theorem.
\begin{restatable}{thm}{technical}
\label{thm:technical}
 There exists a family of small parameters $(\epsilon_j)_{j \in \set{0, \dots, \abs{E}}}$ such that the following holds, for any $\delta$ small enough.
 Let $(\Etilde, \Einf)$ be a valid configuration and $j \coloneqq \lvert\Etilde \setminus \Einf\rvert$. 
 Fix an interval $[\theta_0, \theta_1]$ and a $\delta$-trajectory $\l$.
Let $\lstar$ 
be the equilibrium trajectory for the generalized subnetwork defined by $(\Etilde, \Einf)$ starting with $\lstar(\theta_0)$ being a valid labeling closest to $\l(\theta_0)$, and denote by $T$ be the time required for $\lstar$ to reach steady state. 
 Then supposing that \begin{compactenum}[(i)] 
     \item for every $\linit \in B_{2 \delta}(\l([\theta_0, \theta_1]))$, $E'_{\linit} \subseteq \Etilde$ and $E^*_{\linit} \supseteq \Einf$, and
 \item $\min\{\theta_1 - \theta_0, T\} \leq \Tglob + 1$, 
 \end{compactenum}
 we have 
 \[\norm{\l(\theta) - \lstar(\theta)} \leq \epsilon_j \quad \text{ for all } \theta \in [\theta_0, \theta_1]. \]
\end{restatable}

 A few remarks:
 \begin{itemize}
     \item While $\l(\theta)$ need not be a valid labeling, we can show that it always remains close to a valid labeling. 
 So in this theorem, $\|\lstar(\theta_0) - \l(\theta_0)\| = O(\delta)$. 
     \item Condition (ii), while slightly awkward, will be convenient for inductive purposes. 
        In some cases, we will apply the theorem inductively to an interval of length at most $\Tglob+1$, 
        and in other cases, to a potentially unbounded interval but where the time $T$ is guaranteed to be small.
     \item The equilibrium trajectory $\lstar$ in this theorem is not (in general) an equilibrium of the original network, but rather of the generalized subnetwork determined by $(\Etilde, \Einf)$.
         This is again for inductive purposes; the arcs in $\Einf$ are treated as if they can never empty out, and arcs not in $\Etilde$ are simply not present and cannot be used.
        If we are able to focus on a smaller number of hyperplanes, we can proceed inductively. 
         It may initially seem paradoxical that we show that $\l$ stays close to $\lstar$, if $\lstar$ is not the equilibrium trajectory in the full network that, in the end, we are showing that $\l$ remains close to.
        The resolution is in condition (i), which is very strong.
        At the end of the day, this condition will only hold for intervals where $\lstar$ is close to the equilibrium trajectory of the full network.
     \item This technical theorem implies our main theorem, \Cref{thm:main-full}, fairly immediately.
         Simply take $(\Etilde, \Einf) = (E, \emptyset)$, $\theta_0 = 0$ and $\theta_1$ arbitrarily large.
         The trajectory $\lstar$ is the equilibrium in the original network, starting from the empty network, and so 
         condition (ii) is satisfied by the definition of $\Tglob$.
         Condition (i) is vacuous, and so we obtain the desired claim, with $\epsilon = \epsilon_{|E|}$.
 \end{itemize}
 
 The inductive proof of \Cref{thm:technical} can be broken into two main parts. 
 Unless otherwise indicated, any reference to an equilibrium trajectory (in particular the steady-state direction $\lambda$) refers to such a trajectory in the generalized subnetwork defined by $(\Etilde, \Einf)$, and $\Omega$ refers to the set of valid labels in this generalized subnetwork. 
 
 \paragraph{Part I: Before reaching (near to) steady state.}
 This first part is heavily inductive, and makes little direct use of the properties of $\delta$-trajectories.
 The induction is on $j = \lvert\Etilde \setminus \Einf\rvert$, that is, the number of hyperplanes determining our vector field $X$; see \Cref{fig:figure} for an illustration of these vector fields and some key features of the proof.
 
 Let $\Iloc \subseteq \Omega$ be the steady-state set, and $\lambda$ be the steady-state direction of the generalized subnetwork.

   \begin{figure}[t]
   \centering
   \includegraphics[page=1]{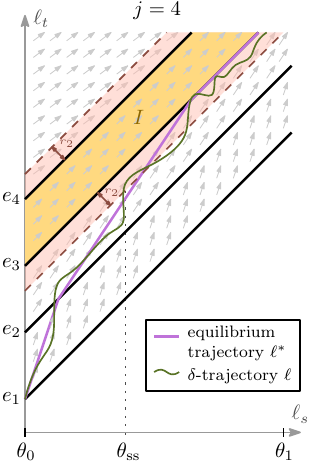}\;\;
   \includegraphics[page=2]{figure}\;\;
   \includegraphics[page=3]{figure}
   \caption{As an example consider a network with only two nodes $s$ and $t$ but four parallel arcs $e_1$ to $e_4$ with transit time $\tau_{e_i}=2i + 1$ and capacities $1$. The network inflow rate is $u_0 = 3$. \emph{On the left} all hyperplanes are present and the equilibrium trajectory reaches steady state as soon as arcs $e_1$, $e_2$ and $e_3$ are active. To prove \Cref{thm:technical} we consider inductively also generalized networks with less arcs. \emph{In the middle} $e_3$ and $e_4$ are removed and therefore $\tilde E = \set{e_1, e_2}$ and $\Einf = \emptyset$. We consider an interval $[\theta_0, \theta_1]$ such that all other hyperplanes keep distance to $\l$. We split the interval at $\thetass$ which is the first point in time $\l$ comes $\r_2$ close to the steady-state set $\Iloc$. Here, we also illustrated the equilibrium trajectory $\lss$, which starts within $B_{\r_3}(\l(\thetass)) \cap \Iloc$ and therefore stays in steady state. $\lstar$ and $\lss$ are close due to the continuity of equilibrium trajectories; see \Cref{thm:continuity_of_Nash}. 
   \emph{On the right} we choose the hyperplanes of $e_2$ and $e_3$, which means that $e_4$ is removed and $e_1$ is promoted to a free arc. Hence $\tilde E = \set{e_1, e_2, e_3}$ and $\Einf = \set{e_1}$.}
   \label{fig:figure}
   \end{figure}

 \medskip
 
 We first consider the behavior when sufficiently far from $\Iloc$.
 A geometric argument shows that this means that $\l(\theta)$ is reasonably far from some hyperplane. 
 Let us sketch this argument.
 If all $j$ hyperplanes do not have a common intersection, then necessarily there is a (network-dependent) lower bound on the distance between hyperplanes, and so $\l(\theta)$ must be ``far'' from at least one hyperplane. 
 Otherwise, if the intersection of all hyperplanes is nonempty, all points in this common intersection can be shown to be part of $\Iloc$.
One can always find a constant $\centerdist$ such that the distance between a given point and the intersection is at most $\centerdist$ times the distance to the farthest hyperplane.
 So being ``far'' from the common intersection of all hyperplanes means being (relatively) far from some hyperplane.
 
 However, we may have a situation where over an interval $\l(\theta)$ remains far from $\Iloc$, but not far from any single hyperplane; rather, we are always far from some hyperplane, but this hyperplane changes over time.
 So we divide the interval into ``periods'', where in each period we are far from a single particular hyperplane; we do this in such a way that each period is not too short.
 We then apply the theorem inductively for each period, one after the other.
 This is a somewhat lossy process; 
 we can control the distance that $\l$ deviates from $\lstar$ over the period in terms of the distance they are apart at the beginning of the period (here the continuity of equilibrium trajectories as stated in \Cref{thm:continuity_of_Nash} is central), but these bounds get worse as we consider more and more periods.
 Fortunately, we can bound the number of periods, because of the fact that $\lstar$ converges to steady state (along with our assumption (i)) gives a bound on the amount of time $\lstar$ (and then inductively, $\l$) stays away from $\Iloc$.
 Since we also argue that each period is not too short, this suffices.
 
In slightly more detail, suppose that $\l(\theta)$ is far from some hyperplane, say the one associated with arc $e'=v'w'$, on an interval $[\theta'_0, \theta'_1]$ of length at most $\Tglob +1$. Therefore, we can proceed inductively on this interval.
 If $e'$ is inactive at $\l(\theta'_0)$, we consider the generalized subnetwork defined by $(\Etilde \setminus \{e'\}, \Einf)$;
 if $e'$ is active (hence has a queue) at $\l(\theta'_0)$, we consider the generalized subnetwork defined by $(\Etilde, \Einf \cup \{e'\})$.
 We apply the theorem inductively to deduce that $\l$ remains close to $\lstarind$, the exact equilibrium of the generalized subnetwork starting from  a closest valid point to $\l(\theta'_0)$. As long as $\l(\theta)$ is further than $\epsilon_{j-1}$ from the hyperplane, then we can in addition deduce that $\lstarind$ does not hit the hyperplane either, and so due to continuity of equilibrium trajectories (\Cref{thm:continuity_of_Nash}) $\lstarind$ stays close to $\lstar$ on this interval, and we have what we want for this particular interval.
 
 
 \medskip
 Denote the first point in time that $\l$ gets within distance $\r_2$ of $\Iloc$ (for some suitable small parameter $\r_2$) by $\thetass$.
The next claim is that $\l$ remains (somewhat) close to $\Iloc$ for the remainder of the evolution:
for some small parameter $\r_3$ (which may be much larger than $\r_2$), 
$\dist(\l(\theta), \Iloc) < \r_3$ for all $\theta \in [\thetass, \theta_1]$.
In order to reach a distance $\r_3$ from $\Iloc$, there will need to be an interval $[\thetastart, \thetaend]$ where $\dist(\l(\thetastart), \Iloc) \leq \r_2$, $\dist(\l(\thetaend), \Iloc) \geq \r_3$, and $\dist(\l(\theta), \Iloc) \geq \r_2$ for all $\theta \in [\thetastart, \thetaend]$.
Since $\l$ remains far from $\Iloc$ in this interval, we can apply what we have already shown to deduce that $\l$ remains close to the equilibrium trajectory $\lstarvar$ starting from a point $\lstarvar(\thetastart)$ close to $\l(\thetastart)$.
But this equilibrium trajectory will reach steady state very quickly, by \Cref{thm:steady-state}.
By choosing $\r_3$ large enough compared to $\r_2$ (but still with $\r_3 \to 0$ as $\delta \to 0$),
we can ensure that this happens by some time $\theta' < \thetaend$.
This exploits that $\delta$-trajectories can be shown to be approximately Lipschitz in a certain sense, and so the interval $[\thetastart, \thetaend]$ cannot be too short if $\r_3$ is large.
But this means that $\l(\theta')$, being close to $\lstarvar(\theta')$, is close to $\Iloc$\,---\,a contradiction.

 \paragraph{Part II: While close to steady state.}
 At this point, we have deduced that $\l(\theta)$ is close to $\lstar(\theta)$ until some time $\thetass$, and that $\l$ remains close to $\Iloc$ from time $\thetass$ forwards.
It remains to argue that $\l(\theta)$ remains within distance $\epsilon_j$ of $\lstar(\theta)$ for all $\theta \geq \thetass$ (for some small parameter $\epsilon_j$). The main part is to prove that $\l$ stays within $\epsilon_j$ distance to an equilibrium trajectory $\lss(\theta) = \lss(\thetass) + (\theta - \thetass)\lambda$ with $\lss(\thetass) \in \Iloc$ close to $\l(\thetass)$ (and therefore close to $\lstar(\thetass)$).
 Notice that we are asking for something quite strong.  
 It is not enough to show that $\l$ moves in ``roughly the right direction''; even a small error in the direction, if maintained, would lead to a large error after a long enough period of time, and we might be considering an arbitrarily long interval. 
 Some amount of self-correction is required; if $\l$ deviates from $\lstar$ in some direction by a significant amount, it should not deviate further in this direction (but could drift away in some other direction).
 
Consider some (possibly large) $\theta \in (\thetass, \theta_1]$, and let $\Delta \theta \coloneqq \theta - \thetass$, and $\Delta \l \coloneqq \l(\theta) - \l(\thetass)$.
 Let us also define $\Delta x_e$ to be the measure of particles that use arc $e=vw$, and enter the arc at some time in the interval $(\l_v(\thetass), \l_v(\theta)]$. 
Observe that if $\l$ was an exact equilibrium, then $(\Delta \l/\Delta \theta, \Delta x / \Delta \theta)$ is a solution to the thin flow equations \eqref{eq:thin} for configuration $(\Etilde, \Einf)$.
 Since thin flows are unique (in terms of $\lambda$)~\CCL,
it follows that $\Delta \l / \Delta \theta = \lambda$.
(If thin flows were not label-unique, then a second distinct solution $(\tilde{y}, \tilde{\lambda})$ with $\tilde{\lambda} \neq \lambda$ would yield a distinct equilibrium trajectory then $\tilde{\l}(\theta) = \tilde\l(\thetass) + \tilde{\lambda}\Delta \theta$ receding from $\lstar$ at a linear rate. Uniqueness of thin flows is thus certainly a necessary fact, though far from sufficient, for our desired convergence claim.)
 
One part of our approach can be viewed as taking the proof of \cite{cominetti2015existence} on the uniqueness of thin flows, and making it ``more robust'' in certain ways. 
 To explain this, we begin by sketching the basic idea of this proof (modified slightly to suit our present purposes).
Suppose for a contradiction that $(\tilde{y}, \tilde{\lambda})$ is a second solution to \eqref{eq:thin} for configuration $(\Etilde, \Einf)$, with $\tilde{\lambda} \neq \lambda$.
Suppose that 
$S := \{ v \in V : \tilde{\lambda}_v/\lambda_v < 1\}$ is nonempty and proper (if it is not, meaning that $\tilde{\lambda}_v \geq \lambda_v$ for all $v$, swap the role of $\lambda$ and $\tilde{\lambda}$, after which $S$ must be proper, given that clearly $s \notin S$).
 One can then make the following key observations, as a consequence of the thin flow equations:
 \begin{itemize}
    \item 
        All arcs $e=vw$ entering $S$ have $\tilde{y}_e \leq y_e$, and the inequality is strict if $\tilde{y}_e > 0$. (Briefly: if $\tilde{y}_e > 0$, then the thin flow equations \eqref{eq:thin} require that $\tilde{\lambda}_w \geq \tilde{\lambda}_v$; 
    since $e$ enters $S$, it follows that $\lambda_w > \lambda_v$, and then the thin flow equations require that $y_e = \lambda_w \nu_e$ and $\tilde{y}_e \leq \tilde{\lambda}_w \nu_e < \lambda_w\nu_e = y_e$.)

    \item All arcs $e=vw$ leaving $S$ have $\tilde{y}_e \geq y_e$, and the inequality is strict if $y_e > 0$. The argument for this is similar to the above.
 \end{itemize}
 Since $y$ and $\tilde{y}$ are both $s$-$t$-flows of the same value, $y(\delta^+(S)) - y(\delta^-(S)) = \tilde{y}(\delta^+(S)) - \tilde{y}(\delta^-(S))$.
This yields an immediate contradiction if $\tilde{y}(\delta^+(S)) > 0$ or $y(\delta^-(S))>0$. 
A small further argument rules out the case that these crossing flows are both zero.
 
 We proceed with a similar cut-based argument in order to reach a contradiction if $\|\Delta \l - \lambda \Delta\theta\|$ is very large.
 In order to do this, we first demonstrate that \emph{some} of the thin flow conditions in \eqref{eq:thin} hold \emph{approximately}. 
Here we directly invoke properties of strict $\delta$-equilibria.
For instance, we are able to show the following (see \Cref{lem:similar_to_thin_flow}):
\begin{itemize}
    \item $\Delta x / \Delta \theta$ is approximately an $s$-$t$-flow of value $u_0$; the appropriate flow conservation constraints hold at each node, up to an $O(\delta)$ error.
    \item For $e=vw \in \Einf$, $|\Delta x_e - \nu_e \Delta \l_w| \leq \nu_e \delta$. 
    \item For $e=vw \in \Etilde \setminus \Einf$, we can show that $\Delta x_e \leq  \nu_e \Delta \l_w+ \nu_e \delta$.
        The thin flow equations imply the exact version of this (without the $\nu_e\delta$ error term), though this is a somewhat weak implication. In particular, we cannot directly show an approximate version of the statement that for $(\tilde y, \tilde{\lambda})$ a thin flow, and $e=vw \in \Etilde$ with $\tilde \lambda_w > \tilde \lambda_v$, $\tilde y_e = \nu_e \tilde{\lambda}_w$.
\end{itemize} 

We then define a cut $(S, V \setminus S)$ based on the ratios $\Delta \l_v/ \lambda_v$, including nodes whose ratio is below some threshold in $S$.
Our goal is then to show that $\Delta x(\delta^+(S)) / \Delta \theta$ is significantly larger than $y(\delta^+(S))$ and $\Delta x(\delta^-(S))/ \Delta \theta$ is significantly smaller than $y(\delta^-(S))$, a contradiction to the fact that $\Delta x /\Delta \theta$ is approximately an $s$-$t$-flow of the same value as $y$.
In order to obtain the desired contradiction, we need to use the above properties, and also some further conclusions that can be drawn from induction.
Significant technical complications arise due to the approximate nature of the information we have on $\l$.



\section{Proof of convergence of \texorpdfstring{$\delta$}{delta}-trajectories}\label{sec:proof}

In this section we prove the technical version of our main theorem, which we repeat below. 
Recall that $\Tglob$ denotes the time an equilibrium trajectory in $G$ starting from the empty network needs to reach steady state.

\technical*

Before proving the theorem, let us first note that our main \Cref{thm:main-full} is an immediately corollary.
\begin{proof}[Proof of \Cref{thm:main-full}]
Choose $\delta$ small enough such that $\epsilon_{\abs{E}} \leq \epsilon$ and consider the interval $[\theta_0, \theta]$ for any $\theta \in \Rplus$. Let $(E, \emptyset)$ be the valid configuration of the original network.
Since $T=\Tglob$ and $\Etilde \setminus \Einf =E$, both conditions of \Cref{thm:technical} are fulfilled.
Thus,
\[\norm{\l(\theta) - \lstar(\theta)} \leq \epsilon_{\abs{E}} \leq \epsilon.\]
\end{proof}

The remainder of this section constitutes the proof of \Cref{thm:technical}.
We will prove the theorem 
by induction on $j$, so
assume that the claim holds for all $j' < j$.
Fix a $\delta$-trajectory $\l$, and an interval $[\theta_0, \theta_1]$ such that conditions (i) and (ii) of \Cref{thm:technical} are fulfilled.
Let $\Omega$ and $\Iloc$ refer respectively to the set of valid labelings and the the steady-state set for the generalized subnetwork defined by $(\Etilde, \Einf)$.

We will define a point in time $\thetass \in [\theta_0, \theta_1]$, such that until that point in time the trajectory has a guaranteed distance to steady state described by some small parameter and at $\thetass$, the trajectory is reasonably close. 
We will first consider in Part I the interval $[\theta_0, \thetass]$, and also show that throughout $[\thetass, \theta_1]$, $\l$ remains close to $\Iloc$.
In Part II we show that $\l$ also remains close to the exact equilibrium trajectory after time $\thetass$.

\subsection{Some useful basic results}
We begin with some basic lemmas that will be useful in various parts of our proof. 
We defer the proofs to the appendix.

The first lemma concerns some basic geometry about the hyperplane arrangements that define the vector field $X(\cdot)$.
For any $e=vw \in E$, let $H_e$ denote the hyperplane determined by edge $e$, namely,
\[ H_e \coloneqq \{ l \in \R^V : l_w - l_v = \tau_e \}. \]
Also, define $H_F \coloneqq \bigcap_{e \in F} H_e$ for any $F \subseteq E$.

We say that a set of edges $F \subseteq E$ is \emph{compatible} if $H_F \neq \emptyset$; that is, all hyperplanes corresponding to edges in $F$ intersect in a nonempty common subspace. 

\begin{restatable}{lemma}{hyperplaneseparation}
\label{lem:hyperplaneseparation}
    There exist constants $\centerdist$ and $\centercompat$ (depending only on the instance) so that the following hold.
    \begin{enumerate} \renewcommand\labelenumi{(\roman{enumi})}\renewcommand\theenumi\labelenumi
        \item 
            For any compatible set $F \subseteq E$, 
            $\dist(\linit, H_F) \leq \centerdist \max_{e \in F} \dist(\linit, H_e)$ for all $\linit \in \R^V$. 
            \label{it:hyperplane_far}
        \item The set of hyperplanes that intersect $B_{\centercompat}(\linit)$ are compatible, for any $\linit \in \R^V$.
        \label{it:hyperplanes_compatible}
    \end{enumerate}
\end{restatable}
\medskip

It is easy to see that for any valid configuration $(E', E^*)$, a solution $(x', \l')$ to the thin flow equations satisfies $\l'_v \leq \TrajLip$ for all $v \in V$, where $\TrajLip := \max\Bigl\{1, u_0/\min_{e \in E} \nu_e\Bigr\}$.
It follows that if $\l$ is an equilibrium trajectory, then $\l$ is $\kappa$-Lipschitz (with respect to the infinity norm).

It is not hard to construct examples of $\delta$-trajectories which are discontinuous. 
Nonetheless, we have the following generalization to $\delta$-trajectories of a weak form of Lipschitz continuity.
 \begin{lemma}\label{lem:deltalip}
     There exist constants $\TrajLipDelta$ and $\jump$ (depending only on the instance) so that the following holds.
     For any $\delta$-trajectory $\l$, and 
     any $\theta_1 < \theta_2$,
     \[ 
         \norm{\l(\theta_2) - \l(\theta_1)} \leq  \TrajLipDelta(\theta_2 - \theta_1) + \jump \delta.
     \]
 \end{lemma}
 Later we prove this for $\epsilon$-equilibria; see \Cref{thm:epsilon_lipschitz}. Since every strict $\delta$-equilibrium is a $\delta$-equilibrium, this lemma follows immediately.

\medskip

If $\l$ is a $\delta$-trajectory, it need \emph{not} to be the case that $\l(\theta) \in \Omega$ for each $\theta$.
The reason is that it can be that an arc $e=vw$ has a queue, but there is no path in $E'_{\l(\theta)}$ from $w$ to $t$; particles at the back of the queue on $e$ at time $\l_v(\theta)$ cannot arrive at time $\l_t(\theta)$, but this is allowable. 
However, this is a minor technical issue, as $\l(\theta)$ will always be close to a point in $\Omega$.
\begin{restatable}{lemma}{almostvalid}
\label{lem:almost-valid}
    Let $\l$ be any $\delta$-trajectory.
    Then 
    $\dist(\l(\theta), \Omega) \leq |V|\delta$
    for any $\theta \in \Rplus$. 
\end{restatable}

\subsection{Part I: Reaching (near to) steady state and staying there}

With help of the small parameter $\r_2$ defined in the following lemma, we can give an exact definition of $\thetass$. Let $\thetass \in [\theta_0, \theta_1]$ be maximal such that $\l(\theta) \notin B_{\r_2}(\Iloc)$ for all $\theta \in [\theta_0, \thetass)$. We will show that $\l$ remains close to $\lstar$ on the interval $[\theta_0, \thetass]$.
The following lemma, applied with $[\theta_2, \theta_3] = [\theta_0, \thetass]$, will do most of the work; we will use it again when showing that we remain close to $\Iloc$ after time $\thetass$.

\begin{lemma} \label{lem:close_outside_of_steady_state_region}
For sufficiently small $\delta$, there exist two small parameters $r_1$ and $\r_2$ with $\r_1 \leq \r_2$ such that the following holds.
Fix any compact interval $[\theta_2, \theta_3]$ with $\theta_3 - \theta_2 \leq \Tglob + 1$ such that 
$\l(\theta) \notin B_{\r_1}(\Iloc)$ for all $\theta \in [\theta_2, \theta_3)$.
Let $\lstarvar$ be the equilibrium trajectory (in the generalized subnetwork defined by $(\Etilde, \Einf)$) starting from a feasible configuration closest to $\l(\theta_2)$.
Then
\[\norm{\l(\theta) - \lstarvar(\theta)} \leq \r_2 \quad \text{ for all } \theta \in [\theta_2, \theta_3).\]
\end{lemma}
\begin{proof}
If $j=0$, then $\Etilde = \Einf$ and thus the vector field $X^{(\Etilde, \Einf)}$ is constant, with every vector in $\R^V$ being a valid labeling.
So $\Iloc = \R^V$, implying that the lemma holds trivially, with $\r_1 = \r_2 = 0$; we must have $\theta_2 = \theta_3$ in order to satisfy the conditions of the lemma, and $\lstarvar(\theta_2) = \l(\theta_2)$. 
So assume $j \geq 1$ in the remainder of the proof.

Let $\Hloc$ be the set of hyperplanes that corresponds to the arcs in $\Eloc \setminus \Einf$ and let $\centerdist$ and $\centercompat$ be  the two constants from \Cref{lem:hyperplaneseparation}.
\begin{claim} \label{claim:hyperplane_faraway}
    Assuming $\delta$ is sufficiently small, 
for every $\xi \in [\theta_2, \theta_3]$ there exists at least one hyperplane $H$ in $\Hloc$ with $\dist(\l(\xi), H) > \tfrac{\r_1}{\centerdist}$. 
\end{claim}
\begin{nestedproof}
    Assume $\delta$ is small enough that $\tfrac{\r_1}{\centerdist} \leq \centercompat$; this is possible, since $\centerdist$ and $\centercompat$ do not depend on $\delta$.
Suppose there exists a point in time $\xi \in [\theta_2, \theta_3]$ for which all hyperplanes in $\Hloc$ individually intersect $B_{\nicefrac{\r_1}{\centerdist}}(\l(\xi))$. In this case \Cref{lem:hyperplaneseparation}~\ref{it:hyperplanes_compatible} states that all hyperplanes in $\Hloc$ are compatible. In other words all these hyperplanes intersect in a nonempty common subspace. 
Observe that in this case this intersection $H_{\Eloc \setminus \Einf}$ is contained in $\Iloc$.
This follows since any point in the intersection of all hyperplanes is feasible for a valid configuration, i.e. $H_{\Eloc \setminus \Einf} \subseteq \Omega$. Thus, at any such point, we can start a trajectory. Such a trajectory is in steady state immediately, since it can move in any direction without ever hitting another hyperplane. Thus, $H_{\Eloc \setminus \Einf} \subseteq \Iloc$.
Moreover, by \Cref{lem:hyperplaneseparation}~\ref{it:hyperplane_far}  the distance of a point to the intersection of some hyperplanes can be upper bounded by using the maximal distance to one of the hyperplanes. For us this yields $\dist(\l(\xi), H_{\Eloc \setminus \Einf}) \leq \Gamma \max_{H \in \Hloc} \dist(\l(\xi), H) \leq r_1$. The last inequality follows by the assumption that all hyperplanes intersect $B_{\nicefrac{\r_1}{\centerdist}}(\ell(\xi))$.
But by the conditions of \Cref{lem:close_outside_of_steady_state_region} $\l(\xi) \notin B_{\r_1}(I) \supseteq B_{\r_1}(H_{\Eloc \setminus \Einf})$.
Thus, there does not exist a point in time $\xi$ such that all hyperplanes intersect  $B_{\nicefrac{\r_1}{\centerdist}}(\l(\xi))$. In other words, there is always a hyperplane $H \in \Hloc$ with $\dist(\l(\xi), H) > \tfrac{\r_1}{\centerdist}$.
\end{nestedproof}

We will later choose $\r_1$ large enough that $\tfrac{\r_1}{\centerdist}$ will be larger than $\epsilon_{j-1}$, which itself is much bigger than $2 \delta$. 
Hence, \Cref{claim:hyperplane_faraway} enables us to use induction by removing the distant hyperplane. However, it might be the case that this hyperplane becomes close or might even be touched or crossed by $\l$ later on. For that reason we split the interval into periods, such that in each period we are far from some particular hyperplane.

For defining these periods let us start with $\thetastart = \theta_2$.
Consider the arc $e = vw \in \Eloc \setminus \Einf$ for which its hyperplane $H_e \in \Hloc$ is furthest from the current point $\l(\thetastart)$. 
This distance is at least $\tfrac{\r_1}{\centerdist}$.
The period lasts until the time where the distance between $\l(\theta)$ and $H_e$ decreases to $\epsilon_{j-1}$ or until $\theta_3$. We denote the end of this period by $\thetaend$ and in the case of $\thetaend < \theta_3$, we start the next period from there.

We can apply induction on each period $[\thetastart, \thetaend]$ by setting $\Etildeind \coloneqq \Eloc \setminus \set{e}$ and $\Einfind = \Einf$ in the case that $\l$ is on the side of $H_e$ for which $e$ is inactive ($\l_w(\thetastart) < \l_v(\thetastart) + \tau_e$) or by setting $\Eloc_\textrm{ind} \coloneqq \Eloc$ and $\Einf_\textrm{ind} = \Einf  \cup \set{e}$ in the case that $\l$ is on the side of $H_e$ for which $e$ is queueing ($\l_w(\thetastart) > \l_v(\thetastart) + \tau_e$). (Note that $\ell$ remains on the same side of the hyperplane $H_e$ for the whole interval $[\thetastart, \thetaend]$.)
Either way
\begin{equation} \label{eq:induction_applied}
\norm{\l(\xi) - \lstarind(\xi)} \leq \epsilon_{j-1} \quad \text{ for all } \xi \in [\thetastart, \thetaend],
\end{equation}
where $\lstarind$ is the equilibrium trajectory in the network defined by $(\Etildeind, \Einfind)$ with the starting point $\lstarind(\thetastart)$ set to be a valid labeling closest to $\l(\thetastart)$.

Note as a crucial observation that even though $\lstarind$ lives in the generalized subnetwork defined by $(\Etildeind, \Einfind)$ it is identical to the equilibrium trajectory starting at the same point but in the original subnetwork (defined by $(\Eloc, \Einf)$) within the interval $[\thetastart, \thetaend)$. This follows from \eqref{eq:induction_applied} along with the fact that $d(\l(\xi), H_e) > \epsilon_{j-1}$ for all $\xi \in [\thetastart, \thetaend)$.

Next we bound the number of periods $k$ within $[\theta_2, \theta_3]$. In each period, except the last one, $\l$ moves at least a distance of $\tfrac{\r_1}{\centerdist} - \epsilon_{j-1}$.
By choosing $\r_1$ big enough, we can assume this to be bigger than $\tfrac{\r_1}{2\centerdist} + \jump \delta$, where $\jump$ is the offset-coefficient defined in \Cref{lem:deltalip}. Therefore \Cref{lem:deltalip} implies that the length of each period is at least 
 $\tfrac{\r_1}{2\centerdist\TrajLipDelta}$. Hence,
\[N_{\r_1} \coloneqq \left\lfloor\frac{(\Tglob + 1) 2\centerdist \TrajLipDelta}{\r_1}\right\rfloor + 1\]
 is an upper bound on $k$.

Finally, we combine all this to get a bound on the maximal distance between $\l(\xi)$ and $\lstarvar(\xi)$ for $\xi \in [\theta_2, \theta_3]$.
To do so we define an increasing sequence of small parameters $(\rsub_i)_{i \in \Znonneg}$ which bound the distance between $\l$ and $\lstarvar$ for all $\theta$ within the $i$-th period, as follows.

At time $\theta_2$, $\l$ and $\lstarvar$ are within a distance of $\delta \abs{V}$ due to \Cref{lem:almost-valid}. Therefore, we set $\rsub_0 \coloneqq \delta \abs{V}$.

Next, consider the $i$-th period $[\thetastart, \thetaend]$ and suppose $\norm{\l(\thetastart) - \lstarvar(\thetastart)} \leq \rsub_{i-1}$.
By \Cref{lem:almost-valid} (and the triangle inequality) it holds that $\norm{\lstarind(\thetastart) - \lstarvar(\thetastart)} \leq \rsub_{i-1} + \delta \abs{V}$. By \Cref{thm:continuity_of_Nash} an equilibrium trajectory depends continuously on its start value. Hence, there exists a small parameter $C_i$ such that
\[\norm{\lstarind(\xi) - \lstarvar(\xi)} \leq C_i \quad \text{ for all } \xi \in [\thetastart, \thetaend].\]
This together with \eqref{eq:induction_applied} and the triangle inequality implies
\[\norm{\l(\xi) - \lstarvar(\xi)} \leq C_i + \epsilon_{j-1} \eqqcolon \rsub_{i} \quad \text{ for all } \xi \in [\thetastart, \thetaend].\]

Note that $\rsub_i$ is defined for all $i \in \Znonneg$ independent of the upper bound on periods $N_{\r_1}$. 
In other words, for $\delta \to 0$ all $\rsub_i$ go to $0$ (independently of $\r_1$) and $N_{\r_1}$ goes to $\infty$. So we can choose $\r_1$ to go to $0$ slowly enough such that 
\begin{equation}
	\label{eq:r2}
\r_2 \coloneqq \max \{\rsub_{N_{\r_1}}, \r_1 \} \tag{r2}
\end{equation}
still goes to $0$ and therefore is a small parameter.
%
%
\end{proof}

%
We would like to apply this lemma with $[\theta_2,\theta_3] = [\theta_0, \thetass]$ and $\lstarvar =\lstar$ in order to deduce that $\l$ stays within distance $\r_2$ of $\lstar$.
However we need to ensure that the condition that $\thetass - \theta_0 \leq \Tglob + 1$ is satisfied.
Rather than showing this directly, we instead apply the lemma to the interval
 $[\theta_0, \thetaend]$ with $\thetaend = \min \set{\thetass, \theta_0 + T}$, where $T$ is the time required for $\lstar$ to reach steady state. 
 By condition (ii) of \Cref{thm:technical}, we have that $\min \set{\theta_1 - \theta_0, T} \leq \Tglob + 1$; so the interval $[\theta_0, \thetaend]$ does satisfy the condition of the lemma.
It remains to show that $\thetaend = \thetass$.

Suppose not, meaning that $\thetaend < \thetass$.
By the definition of $T$, $\lstar(\thetaend) \in \Iloc$, and $\norm{\l(\thetaend), \lstar(\thetaend)} \leq \r_2$ by the lemma.
But this implies $\dist(\l(\thetaend), \Iloc) \leq \r_2)$, giving us the contradiction that $\thetass \geq \thetaend$.

So at this point, we have by \Cref{lem:close_outside_of_steady_state_region} that 
\[\norm{\l(\theta) - \lstar(\theta)} \leq \r_2 \leq \epsilon_j \quad \text{ for all } \theta \in [\theta_0, \thetass].\]
Next, we show in \Cref{lem:staying_close_to_steady_state_region} that we will stay close to the steady-state set for the remaining interval $[\thetass, \theta_1]$.
 

\begin{lemma}[Staying close to $\Iloc$] \label{lem:staying_close_to_steady_state_region}
    There exists a small parameter $\r_3$ such that:
\[\l(\theta) \in B_{\r_3}(\Iloc) \quad \text{ for all } \theta \in [\thetass, \theta_1].\]
\end{lemma}
\begin{proof}
Let $T^* = \T \cdot (\r_2 + \jump\delta + \abs{V}\delta)$ be the maximal time an equilibrium trajectory starting within distance $(\r_2 + \jump\delta + \abs{V}\delta)$ from $\Iloc$ needs to reach steady state; see \Cref{thm:steady-state}.
We choose 
\begin{equation}
\label{eq:r3}
\r_3 \coloneqq \r_2 + T^* \TrajLipDelta + 2\jump\delta, \tag{r3}
\end{equation}
where $\TrajLipDelta$ and $\jump$ are the constants from \Cref{lem:deltalip}.
Since $\r_2$ and hence $T^*$ are small parameters, $\r_3$ is too.

Suppose $\l$ leaves $B_{\r_3}(\Iloc)$ at any point in time within $[\thetass, \theta_1]$. Since $\l(\thetass) \in B_{\r_2 + \jump\delta}(\Iloc)$ there exists a time interval $[\thetastart, \thetaend]$ with $\l(\theta) \notin B_{\r_2}(\Iloc)$ for all $\theta \in (\thetastart, \thetaend]$ 
on which $\l$ covers at least a distance of $\r_3 - (\r_2 + \jump\delta) = T^* \TrajLipDelta + \jump\delta$ (as it moves from the boundary of $B_{\r_2 + \jump\delta}(\Iloc)$ to outside of $B_{\r_3}(\Iloc)$).
 
By \Cref{lem:deltalip} this interval has length at least $T^*$.
To obtain a contradiction consider the equilibrium trajectory $\lstarvar$ starting with $\lstarvar(\thetastart)$ as close as possible to $\l(\thetastart)$.
By \Cref{lem:almost-valid}, $\dist(\lstarvar(\thetastart), \l(\thetastart)) \leq \abs{V} \delta$. 
Hence, $\dist(\lstarvar(\thetastart), \Iloc) \leq \r_2 + \jump\delta + \abs{V} \delta$ which implies that $\lstarvar$ reaches steady state within time $T^*$.

But applying \Cref{lem:close_outside_of_steady_state_region} to the interval $[\thetastart, \thetastart + T^*]$ (note that by taking $\delta$ sufficiently small, this is shorter than $\Tglob + 1$) we obtain that
\[\norm{\l(\thetastart + T^*) - \lstarvar(\thetastart + T^*)} \leq \r_2.\]
This is a contradiction as $\thetastart + T^* \in [\thetastart,\thetaend]$ and thus
$\l(\thetastart + T^*) \notin B_{\r_2}(\Iloc)$ but $\lstarvar(\thetastart + T^*) \in \Iloc$.
\end{proof}

\subsection{Part II: Staying close to the equilibrium trajectory while close to steady state}
In this section we will prove the following lemma, which says that we stay close to the equilibrium trajectory $\lss$ starting from a nearby point in $\Iloc$.

\begin{lemma} \label{thm:close_to_Nash}
 For sufficiently small $\delta$ , there exists a small parameter $\reight$ such that
\[ \|\l(\theta) - \lss(\theta)\| \leq \reight \quad \text{ for all } \theta \in [\thetass, \theta_1],\]
where $\lss$ is an equilibrium trajectory starting from a point $\lss(\thetass) \in B_{\r_3}(\l(\thetass)) \cap \Iloc$.
\end{lemma}

With this lemma in hand, completing the proof of \Cref{thm:technical} will be very easy: $\l$, $\lstar$ and $\lss$ are close to each other at time $\thetass$, and by continuity of dynamic equilibria (\Cref{thm:continuity_of_Nash}), $\lstar$ must stay close to $\lss$, and hence $\l$.
We start with various definitions and auxiliary lemmas that we will need in the proof of \Cref{thm:close_to_Nash}.  

Since $\lss$ starts in the steady-state set, $\lss(\theta) = \lss(\thetass) + \lambda(\theta - \thetass)$ for all $\theta \geq \thetass$, with $\lambda$ denoting the steady-state direction.
Let $y$ be a corresponding $s$-$t$-flow so that $(y, \lambda)$ is a solution to the thin flow equations \eqref{eq:thin} for the configuration $(\Eloc, \Einf)$.

For any arc $e=vw \in E$ and all points in time $\theta$ we define
\[\lx_e(\theta) \coloneqq \mu(\{a : e \in \P(a) \text{ and } \d_v(a) \leq \l_v(\theta) \}). \]
This means $\lx_e(\theta)$ denotes the cumulative flow that entered arc $e$ until time $\l_v(\theta)$.
The focus is on the behavior of $\l$ after time $\thetass$. For that reason we define
\[\Delta \theta \coloneqq \theta - \thetass, \qquad \Delta \l(\theta) \coloneqq \l(\theta) - \l(\thetass), \quad \text{ and } \quad \Delta \lx(\theta) \coloneqq \lx(\theta) - \lx(\thetass)\]
for any given $\theta \in [\thetass, \theta_1]$.
For the sake of clarity we write $\Delta \l$, $\Delta \lx$ instead of $\Delta \l(\theta)$, $\Delta \lx(\theta)$ whenever the choice of $\theta$ is unambiguous.

The following crucial claim says that $(\Delta \lx / \Delta \theta, \Delta \l / \Delta \theta)$ \emph{approximately} satisfies \emph{some} of the thin flow equations \eqref{eq:thin} for configuration $(\Etilde, \Einf)$ which have $(y,\lambda)$ as their exact solution. If $(\Delta \lx / \Delta \theta, \Delta \l / \Delta \theta)$ satisfied all the thin flow equations exactly, we would already have (using label-uniqueness of the thin flow equations) that $\Delta \l =  \lambda \Delta \theta$, meaning that the trajectories $\l(\theta) $ and $ \lss(\theta)$ move in parallel from the start point on, i.e. their distance can be bounded by $\r_3$.  

\begin{lemma}
\label{lem:similar_to_thin_flow}
Fix $\theta \in [\thetass, \theta_1]$. It holds that:
\begin{enumerate} \renewcommand\labelenumi{(\roman{enumi})}\renewcommand\theenumi\labelenumi
    \item For $e =vw \in \Eloc$, $\Delta \lx_e \leq \nu_e \Delta \l_w + \nu_e \delta$. \label{it:Delta_x_upper_bound}
    \item For $e =vw \in \Einf$, $\abs{\Delta \lx_e - \nu_e \Delta \l_w} \leq \nu_e \delta$. \label{it:Delta_x_lower_bound}
    \item For $e \notin \Eloc$, $\Delta \lx_e = 0$. \label{it:Delta_x_inactive}
    \item $\Delta \lx$ is approximately an $s$-$t$-flow of value $u_0 \cdot \Delta \theta$: \label{it:Delta_x_flow}
        \[\abs{\sum_{e \in \delta^-(v)} \Delta \lx_e - \sum_{e \in \delta^+(v)} \Delta \lx_e  + u_0\Delta \theta \bigl(\ind_s(v) - \ind_t(v)\bigr)} \leq 2 \delta \nu_\Sigma \quad \text{ for all } v \in V.\] 
    Here, $\ind_w(v) = 1$ if $v = w$ and $0$ otherwise, and $\nu_\Sigma \coloneqq \sum_{e \in E} \nu_e + u_0$.
\end{enumerate}
\end{lemma}
It is helpful to compare with \eqref{eq:thin} (for the configuration $(\Eloc, \Einf)$). 
\ref{it:Delta_x_inactive} and \ref{it:Delta_x_flow} together say that $\Delta x / \Delta \theta$ is approximately an $s$-$t$-flow of the correct value supported on $\Eloc$.
Consider an arc $e=vw \in \Einf$; then the thin flow equations say that $\lambda_w = y_e / \nu_e$, and \ref{it:Delta_x_lower_bound} is an approximate version of this statement.
For an arc $e=vw \in \Eloc \setminus \Einf$, the thin flow equations say that $\lambda_w = \max\{ \lambda_v, y_e/\nu_e\}$. 
This we do not have a full replacement for, but \ref{it:Delta_x_upper_bound} is 
an approximate version of the weaker implication that $y_e \leq \nu_e \lambda_w$.
\begin{proof} First, recall that by definition of shortest path labels it holds in general that
\[ \l_w(\xi)  \leq \l_v(\xi) + \tau_e + \lq_e(\xi) \quad \text{ for all } \xi \in \Rplus.\]
We show each statement individually.
\begin{enumerate}\renewcommand\labelenumi{(\roman{enumi})}\renewcommand\theenumi\labelenumi
\item 
We now define a function $\dx_e$ which maps any agent $a$ for which $e \in \P(a)$ to the mass of particles that traverse $e$ ahead of $a$.
Here, we need to take care of tiebreaking if there is a mass departure from $v$.
We can write
\[ \dx_e(a) \coloneqq \mu(\{ a' : e \in \P(a'),  \d_v(a') \leq \d_v(a), \text{ and either } \d_v(a') < \d_v(a) \text{ or } \entry{a'} \leq \entry{a} \}). \]
The tiebreaking rule is thus that agents with the same departure time from $v$ as $a$ are ahead if their entrance time into the network is smaller.
Remember that $\dq_e(a)$ is defined as the waiting time experienced by $a$ on $e$,  see \eqref{eq:queue-waiting}.

Now we want to, morally speaking, choose the ``last agent'' that entered arc $e$ by time $\l_v(\theta)$.
There does not need to be such a last agent, so we employ a limiting argument. 
Let 
\[A \coloneqq \set{a' \in \agents: \d_v(a') \leq \l_v(\theta), e \in \P(a')}, \]
Note that if $A = \emptyset$, no agents have entered the arc, and so $\Delta \lx(\theta) = 0$, which shows \ref{it:Delta_x_upper_bound} immediately, since $\Delta \l_w \geq 0$. So suppose $A \neq \emptyset$.

For any $\epsilon > 0$, we can choose an agent $\aepsilon \in A$ so that 
\begin{multline*}
    \mu(\{ a' \in \agents : e \in \P(a'), \d_v(\aepsilon) < \d_v(a') < \l_v(\theta)\} \cup \\
    \{a' \in \agents: e \in \P(a'), \d_v(a') = \d_v(\aepsilon) \text{ and } \entry{a'} > \entry{\aepsilon} \}) < \epsilon.
\end{multline*}

For agent $\aepsilon$, we can write $\nu_e\dq_e(\aepsilon)$ as the queue volume at time $\l_v(\thetass)$, plus the cumulative mass of particles that have entered $e$ in front of $\aepsilon$ but after $\thetass$, minus the cumulative mass that have departed the queue on $e$ in the time $[\l_v(\thetass), \d_v(\aepsilon)]$.
Due to the capacity constraint on the arc, the mass that has left $e$ during $[\l_v(\thetass), \d_v(\aepsilon)]$ is not larger than $\nu_e(d_v(\aepsilon) - \l_v(\thetass))$. Hence,
\begin{equation}
    \nu_e\dq_e(\aepsilon) \geq \nu_e\lq_e(\thetass) + \dx_e(\aepsilon) -\lx_e(\thetass) - \nu_e(\d_v(\aepsilon) -  \l_v(\thetass)). \label{eq:dqlb}
\end{equation}

By the choice of agent $\aepsilon$ it holds that $\dx_e(\aepsilon) \geq \lx_e(\theta) - \epsilon$.
Then
\begin{align*}
    \nu_e \l_w(\theta) &\geq \nu_e \l_w(\entry{\aepsilon}) &&\text{($\entry{\aepsilon} \leq \theta$ by the def.\ of $A$)}\\
                &\geq \nu_e(\d_w(\aepsilon) - \delta) &&\text{($\l$ is a $\delta$-trajectory)}\\
                &\geq \nu_e(\d_v(\aepsilon) + \dq_e(\aepsilon) + \tau_e - \delta)&&\text{($e \in \P(\aepsilon)$)}\\
                &\geq \dx_e(\aepsilon) - \lx_e(\thetass)  + \nu_e(\l_v(\thetass) + \lq_e(\thetass) + \tau_e) - \nu_e\delta \;\;  &&\text{(\cref{eq:dqlb})} \\
                &\geq \lx_e(\theta) -  \epsilon - \lx_e(\thetass) + \nu_e\l_w(\thetass) - \nu_e\delta && \text{(def.\ of earliest arrival labels)}\\
                &\geq \Delta\lx_e(\theta) +\nu_e\l_w(\thetass) -\nu_e\delta - \epsilon.&&
\end{align*}
Taking $\epsilon \to 0$, we deduce the desired inequality.

\item Let $e = vw \in \Einf$, i.e., an arc on which $\l$ has a positive queue during the whole interval. This holds since, by assumption of \Cref{thm:technical}, the $\delta$-trajectory $\l$ keeps a minimum distance of more than $2\delta$ from each hyperplane corresponding to some $e \in \Einf$ and thus maintains a queue of length (waiting time) more than $2\delta$.

Moreover, we can again express the queue mass $\nu_e\lq_e(\theta)$ by the flow mass at time $\l_v(\thetass)$, plus the flow entering $e$ within $[\l_v(\thetass), \l_v(\theta)]$, less the flow leaving within this interval.
As $e \in \Einf$, the mass that has left during this interval is precisely $\nu_e(\l_v(\theta) - \l_v(\thetass))$. Formally, we obtain
\begin{equation} \label{eq:equb}
\begin{aligned}
    \nu_e\lq_e(\theta) &= \nu_e \lq_e(\thetass)  + \lx_e(\theta) -\lx_e(\thetass )- \nu_e(\l_v(\theta) - \l_v(\thetass)) \\
                &=\nu_e\lq_e(\thetass) + \Delta \lx_e(\theta) - \nu_e\l_v(\theta) + \nu_e\l_v(\thetass).
\end{aligned}
\end{equation}

At time $\thetass$ there is a positive queue on $e$ with waiting time of more than $2\delta$, i.e., $q_e(\thetass)> 2\delta$.
Hence the following agent set is non-empty (it contains the agents that leave the queue in the interval $[\l_v(\thetass), \l_v(\thetass)+\delta]$, since they reach the end of the arc in $[\l_v(\thetass)+\tau_e, \l_v(\thetass)+\tau_e+\delta]$ and additionally wait at most $\delta$ before leaving $w$.)
\[A \coloneqq \set{a' \in \agents \colon d_w(a') \leq \l_v(\thetass) + \tau_e + \lq_e(\thetass), e \in \P(a')}.\]
It holds that $\sup_{a' \in A} d_w(a') = \l_v(\thetass) + \tau_e + \lq_e(\thetass)$ as a probe particle entering at time $\l_v(\thetass)$ will catch up with agents on the back of the queue and they will leave the arc almost at the same time.
(One would expect that the $\sup$ could be a $\max$ and we can just pick the agent that is last in line at time $\l_v(\thetass)$. However, it could be that there is no such agent, as all agents that would have been there, took another path. In some sense the queue might be open at its end.)
Hence, for any small $\epsilon > 0$ we can pick an agent $\aepsilon \in A$ with 
\begin{equation} \label{eq:last_agent_in_queue}
d_w(a) \geq \l_v(\thetass) + \tau_e + \lq_e(\thetass) - \epsilon.
\end{equation}
Since $\l$ is a $\delta$-trajectory it holds that $d_w(a) \leq \l_w(\entry{a}) + \delta$.
Moreover, $a$ entered arc $e$ not after $\l_v(\thetass)$ (otherwise it would have entered $v$ strictly later than a probe particle but arrive at $w$ earlier, which is impossible.) Hence $\entry{a} \leq \thetass$.

Thus,
\begin{align*}
\nu_e\l_w(\theta) &\leq \nu_e\l_v(\theta) + \nu_e(\tau_e + \lq_e(\theta)) &&\text{(definition of earliest arrival labels)}\\
             &= \Delta \lx_e(\theta) + \nu_e(\l_v(\thetass) +\tau_e + \lq_e(\thetass)) &&\text{(\cref{eq:equb})}\\
             &\leq \Delta\lx_e(\theta) + \nu_e(d_w(a) + \epsilon) &&\text{(\cref{eq:last_agent_in_queue})}\\
             &\leq \Delta\lx_e(\theta) + \nu_e(\l_w(\entry{a}) + \delta + \epsilon) &&\text{($\l$ is a $\delta$-trajectory)}\\
             &\leq \Delta \lx_e(\theta) + \nu_e(\l_w(\thetass) + \delta + \epsilon. &&\text{($\entry{a} \leq \thetass)$}
\end{align*}
Again taking $\epsilon \to 0$, we obtain the desired lower bound on $\Delta \lx_e(\theta)$. Together with \ref{it:Delta_x_upper_bound} this proves \ref{it:Delta_x_lower_bound}.

\item Let $e$ be an arc that is not in $\Eloc$. It holds that $\l_w(\xi) < \l_v(\xi) + \tau_e - 2\delta$ for all $\xi \in [\thetass, \theta]$ as by the requirement of \Cref{thm:technical} the hyperplane corresponding to $e$ has distance of more than $2 \delta$.
Consider the following agent set
\[A \coloneqq \set{a' \in \agents \colon d_v(a') \in [\l_v(\thetass), \l_v(\theta)], e \in \P(a')}.\]
If this agent set is empty that means that $\Delta \lx_e(\theta) = 0$. So suppose there is an agent $a \in A$. Let $\theta' \coloneqq \max\set{\entry{a}, \thetass}$.
Clearly, $d_v(a) \geq \l_v(\theta')$. We obtain
\[d_w(a) \geq d_v(a) + \tau_e + \dq_e(a) \geq \l_v(\theta') + \tau_e > \l_w(\theta') + 2\delta \geq \l_w(\entry{a}) + 2\delta.\]
This is a contradiction to the $\delta$-trajectory property.

\item
    For a vector $f \in \Rplus^{\Eloc}$ and a node $v \in V$, let $\nabla_v f$ denote the net flow at $v$; 
    \[ \nabla_v f := \sum_{e \in \delta^-(v)} f_e - \sum_{e \in \delta^+(v)} f_e. \]
    We will show that 
    \[ 0 \leq \nabla_v x(\theta) + u_0\theta\cdot \ind_s(v) \leq 2\delta \capsum \qquad \text{for any } \theta \geq \thetass \text{ and } v \neq t; \]
    and also that 
\[ -  2\delta \capsum \leq \nabla_t x(\theta) - u_0 \theta\cdot \ind_t(v) \leq 0 \qquad \text{for any } \theta \geq \thetass. \]
    Applying this as well for $\theta=\thetass$ and subtracting yields the claimed bounds on the net flow of $\Delta x$ at any node.

    Let 
 \[ K_v := \bigl\{ a \in \agents_v : \entry{a} \leq \theta \text{ and } d_v(a) > \l_v(\theta) \bigr\}. \]
    We will relate $\nabla_v x(\theta)$ to $\mu(K_v)$ for each $v$.
    
    First consider $v \in V \setminus \{s,t\}$.
    Let $\agents_v$ be the set of agents whose path $\P(a)$ includes node $v$, and
    for $a \in \agents_v$, let $u(a)$ be the node in $\P(a)$ appearing just before $v$.
    Considering what it means for an agent to contribute to the measure defining $x_e(\theta)$ for an arc $e$ entering $v$ but not to an arc leaving $v$, and vice versa, we have that $\nabla_v x(\theta) = \mu(K_v^+) - \mu(K^-_v)$, where
    \begin{align*} 
        K^+_v &:= \bigl\{ a \in \agents_v : d_{u(a)}(a) \leq \l_{u(a)}(\theta) \text{ and } d_v(a) > \l_v(\theta) \bigr\}\\
        K^-_v &:= \bigl\{ a \in \agents_v : d_{u(a)}(a) > \l_{u(a)}(\theta) \text{ and } d_v(a) \leq \l_v(\theta) \bigr\}.
    \end{align*}
    First, we argue that $\mu(K^-_v) = 0$. 
    Fix any $u$ such that $uv$ is an arc, and consider $L_u := \{ a \in K^-_v : u(a) = u\}$.
    A hypothetical particle departing $u$ at time $\l_{u}(\theta)$ arrives at $v$ at time no earlier than $\l_v(\theta)$, so all agents in $L_u$ must have an arrival time as well as departure time at $v$ of $\l_v(\theta)$.
    But then $\mu(L_u) = 0$; a positive measure of agents cannot all ``catch up'' to this hypothetical particle on the arc $uv$.
    This holds for all choices of $u$, and so $\mu(K^-_v) = 0$.

    Each $a \in K^+_v$ satisfies 
    $d_{u(a)}(a) \leq \ell_{u(a)}(\theta)$, and hence $\entry{a} \leq \theta$;
    it follows that $K^+_v \subseteq K_v$.
    So $0 \leq \nabla_v x(\theta) \leq \mu(K_v)$.

    For $v=s$, we have 
    $\nabla_s x(\theta) = -\mu(\{ a \in \agents : d_s(a) \leq \theta \bigr\})$; all agents that depart the source by time $\l_s(\theta) = \theta)$ contribute to the net flow leaving $s$.
  Thus, we can rewrite as $\nabla_s x(\theta) = \mu(K_s) - u_0 \theta$.

    Finally, for $v=t$, we have 
    $\nabla_t x(\theta) = \mu(K'_t)$ where
    \[ K'_t := \{ a \in \agents: d_{u(a)}(a) \leq \l_{u(a)}(\theta)\}. \]
    Since $K'_t \subseteq \{ a \in \agents: \entry{a} \leq \theta\}$, we can rewrite this as
    \[ \nabla_t x(\theta) = u_0\theta - \mu(\{ a \in \agents: \entry{a} \leq \theta \text{ and } d_{u(a)}(a) > \l_{u(a)}(\theta)\}). \]
    Now just as in the argument above, the measure of agents $a \in \agents$ for which $d_{u(a)}(a) > \l_{u(a)}(\theta)$ but $d_t(a) \leq \l_t(\theta)$ is zero. 
So $u_0\theta - \mu(K_t) \leq \nabla_t x(\theta) \leq u_0\theta$.
    
    \medskip

    So all that remains is to bound $\mu(K_v)$ for each $v \in V$.
    For each $a \in K_v$, $d_v(a) \leq \l_v(\entry{a}) + \delta \leq \l_v(\theta) + \delta$, using that $\l$ is a $\delta$-equilibrium.
    So $d_v(a) \in [\l_v(\theta), \l_v(\theta) + \delta]$ for all $a \in K_v$.
    This implies that the arrival time at $v$ of all agents in $K_v$ lies in the interval $[\l_v(\theta) - \delta, \l_v(\theta) + \delta]$, again using that $\l$ is a $\delta$-equilibrium and hence that no agent can wait at $v$ for longer than $\delta$.
    The measure of agents that can arrive in this interval is at most $2\delta \capsum$, giving us the required bound $\mu(K_v) \leq 2\delta \capsum$.

\end{enumerate}
\end{proof}

In the following we compare $\Delta \l / \Delta \theta$ to the steady-state direction $\lambda$. 
%
%
We partition the edge set $\tilde E \setminus \Einf$ into the following three sets:
\begin{align*}
E^= &\coloneqq \set{vw \in \tilde E \setminus \Einf | \lambda_w = \lambda_v},\\
E^> &\coloneqq \set{vw \in \tilde E \setminus \Einf | \lambda_w > \lambda_v},\\ 
E^< &\coloneqq \set{vw \in \tilde E \setminus \Einf| \lambda_w < \lambda_v}.
\end{align*}

We will require the following lemma shown in \cite{OSV21} (where it was shown in a slightly stronger form).

\begin{lemma}[Lemma 4.3, \cite{OSV21}]
\label{lem:Ebigger_and_Esmaller_in_ss}
For any $l \in \Iloc$, 
\begin{enumerate}[(i)]
	\item every arc  in $E^>$ is active, i.e., $e \in E_{l}'$ for all $e \in E^>$; and 
	\item no arc in $E^<$ has a queue, i.e., $e \not\in   E^*_{l}$ for all $e \in E^<$.
\end{enumerate}
\end{lemma}

Let $h: \Rplus \to \Rplus$ be such that any two equilibrium trajectories starting within distance $r$ stay within distance $h(r)$ for all time, and with $h(r) \to 0$ as $r \to 0$. Such a function exists by \Cref{thm:continuity_of_Nash}.

Now we define some more small parameters, including the value of the small parameter $\reight$ claimed in \Cref{thm:close_to_Nash}:
\begin{align}
    \r_4 &\coloneqq h(\r_3 + |V|\delta), \tag{r4} \label{eq:r4}\\
\label{eq:r5}
    \r_5 &\coloneqq 2\r_4 + 2 \r_3+  5\epsilon_{j-1} + \modified{2}\jump \delta, \tag{r5}\\
	\label{eq:r6}
	\r_6 &\coloneqq \max\Set{\frac{4 \r_5}{\lambdamin} \cdot \abs{V}, \frac{3 \delta \nu_\Sigma}{\capmin\lambdamin} \abs{V}^2}, \tag{r6}\\
	\label{eq:r7}
	\r_7 &\coloneqq \r_6 + 5\tfrac{\r_5}{\lambdadiff}, \tag{r7}\\
        \rnew &\coloneqq \r_4 + \epsilon_{j-1} + 2\r_3 + (\TrajLip + \TrajLipDelta) \r_7 + \delta \jump, \tag{r8} \label{eq:rnew}\\
	\label{eq:r8}
        \reight &\coloneqq \max\bigl\{h(\rnew + \delta|V|) + \epsilon_{j-1}, \r_3 + \lambdamax \r_6\bigr\}.\tag{r9}
\end{align}
Here, $\lambdamin \coloneqq \min_{v \in V} \lambda_v$, $\lambdamax \coloneqq \max_{v \in V} \lambda_v$, $\lambdadiff \coloneqq \min \set{\abs{\lambda_v - \lambda_w} \colon v,w \in V \text{ with } \lambda_v \neq \lambda_w}$, and $\capmin \coloneqq \min (u_0, \min_{e \in E}\nu_e)$.  Note that in general for thin flows all labels $\lambda_v$ are positive (see \cite[Lemma 4.2]{OSV21}) and therefore $\lambdamin > 0$.

\begin{definition}
    We say that the trajectory $\l$ is \emphdef{central near $\thetass$} if
\[\abs{\l_w(\theta)-\l_v(\theta)-\tau_e} \leq \r_5  \quad \text{ for all } e = vw \in \Etilde \setminus \Einf \text{ and for all } \theta \in [\thetass, \thetass+r_7].\] 
\end{definition}


This essentially says that within some small initial interval starting at $\thetass$, the trajectory $\l$ remains close to all hyperplanes corresponding to $\Etilde \setminus \Einf$. 
We will proceed differently depending on whether $\l$ is central near $\thetass$.

If $\l$ is not central near $\thetass$, we will rely on induction. 
We are, at some moment very close to $\thetass$, somewhat far from some hyperplane. 
We will argue that the steady-state direction in the smaller instance obtained by dropping this hyperplane is still $\lambda$, and that this steady-state direction takes us further and further from this hyperplane, or keeps the same distance.
Induction can then be applied straightforwardly.

If $\l$ is central near $\thetass$, we have to do much more work; this is the key case where we cannot simply resort to induction.
(Note that in the base case $j=0$ where there are no hyperplanes, $\l$ is vacuously central near $\thetass$.)
\nnote{More explanation; what do we win by knowing that we are central near $\thetass$? But don't get into the robust uniqueness argument yet, leave that description until later. Update: downgrading from a TODO, not crucial for the arxiv version.}
\laura{Intuitively, I have no idea. Technically, we use this fact only at one place, namely we use that queues at the start are all small/ arcs are only a bit inactive. this allows to give some upper bound on the queue length at time $\theta$ / at how much arcs are inactive (and in the robustness argument we use upper bound to deduce that the arcs entering/ leaving this set $S$ have a certain type) }

The following simple lemma just says that $\l$ and $\lss$ cannot get too far apart too quickly.
\begin{lemma}
	\label{lem:close_for_a_bit}
For all $\theta \in [\thetass, \thetass + \r_7]$, 
\[ 	\| \l(\theta) - \lss(\theta)\| \leq \rnew . \]
\end{lemma}
\begin{proof}
    We use $\TrajLip$-Lipschitzness of equilibrium trajectories and approximate Lipschitzness of $\delta$-trajectories (\Cref{lem:deltalip}) to obtain: 
\begin{equation} \label{eq:close_for_small_theta}
	\begin{aligned}
        \norm{\l(\theta) -\lss(\theta)} &\leq \norm{\l(\thetass) -\lss(\thetass)} + \TrajLipDelta (\theta -\thetass) + \delta\jump +\TrajLip(\theta -\thetass)\\
		&\leq \r_3 + \TrajLipDelta \r_7 + \delta \jump + \TrajLip \r_7 \underbrace{\leq}_{\eqref{eq:rnew}} \rnew.
	\end{aligned}
\end{equation}
\end{proof}

\paragraph{Case 1: $\l$ is not central near $\thetass$.}

\begin{lemma}\label{lem:notcentral}
If $\l$ is not central near $\thetass$, and assuming that $\delta$ is sufficiently small, 
\[ \| \l(\theta) - \lss(\theta)\| \leq \reight  \qquad \text{ for all } \theta \in [\thetass, \theta_1]. \]
\end{lemma}
\begin{proof}
Let $\xi_0 \in [\thetass, \thetass + \r_7]$ and $e=vw \in \Etilde \setminus \Einf$ be such that $\abs{\l_w(\xi_0)-\l_v(\xi_0)-\tau_e} > \r_5$.
Let $\xiend$ be the maximal value in $[\xi_0, \theta_1]$ such that 
\begin{equation}\label{eq:xiend}
    \abs{\l_w(\xi) - \l_v(\xi) - \tau_e} > 2\epsilon_{j-1} \text{ for all } \xi \in [\xi_0, \xiend).
\end{equation}
Later we will argue that $\xiend = \theta_1$.
We apply induction on $[\xi_0, \xiend]$ by setting $\Etildeind \coloneqq \Etilde$ and $\Einfind \coloneqq \Einf \cup \set{e}$ if 
$e \in E'_{\l(\xi_0)}$, or $\Etildeind \coloneqq \Etilde \setminus \set{e}$ and $\Einfind \coloneqq \Einf$ otherwise.
(Informally, we move the hyperplane $H_e$ off to infinity, keeping $\l(\xi_0)$ on the same side of the hyperplane.)
Let $\lstarind$ be an equilibrium trajectory in the generalized subnetwork corresponding to $(\Etildeind, \Einfind)$ that starts at a valid labelling as close to $\l(\xi_0)$ as possible;
we know that $\|\lstarind(\xi_0) - \l(\xi_0)\| \leq \delta \abs{V}$ by \Cref{lem:almost-valid}.
As the time to reach steady state for $\lstarind$ is bounded by $T_{\r_3} := \T \cdot (\r_3 + \delta \abs{V})$ (\Cref{thm:steady-state}).
By taking $\delta$ sufficiently small, we can assume that $T_{\r_3} \leq 1$, ensuring that the requirements of \Cref{thm:technical} are fulfilled for the interval $[\xi_0, \xiend]$ (in particular, $\min\{\xiend - \xi_0, T_{\r_3}\} \leq 1 \leq \Tglob +1$).
Hence $\lstarind$ stays within distance $\epsilon_{j-1}$ of $\l$. 
So in order to argue that $\l$ remains close to $\lss$ on $[\xi_0, \xiend]$, it suffices to show that $\lss$ remains close to $\lstarind$. 

Even though $\lstarind$ is defined to be an equilibrium trajectory in the reduced network determined by $(\Etildeind, \Einfind)$, it is also a valid equilibrium trajectory in the original network corresponding to $(\Etilde, \Einf)$ as it does not touch or cross the hyperplane $H_e$.
By \Cref{lem:close_for_a_bit}, $\|\l(\xi_0) - \lss(\xi_0)\| \leq \rnew$, and since $\|\l(\xi_0) - \lstarind(\xi_0)\| \leq \delta|V|$, $\|\lss(\xi_0) - \lstarind(\xi_0)\| \leq \rnew + \delta|V|$.
Since $\lss$ and $\lstarind$ are both equilibrium trajectories on $[\xi_0, \xiend]$,
$\|\lss(\theta) - \lstarind(\theta)\| \leq h(\rnew + \delta|V|)$ for all $\theta \in [\xi_0, \xiend]$.
Applying the triangle inequality and the definition of $\reight$ we have
\[ \|\l(\theta) - \lss(\theta)\| \leq \|\l(\theta) - \lstarind(\theta)\| + \|\lstarind(\theta) - \lss(\theta)\| \leq h(\rnew + \delta|V|) + \epsilon_{j-1} \leq \reight. \]

\medskip

All that remains is to show that $\xiend = \theta_1$.
Let $\lstarss$ be an equilibrium trajectory starting at time $\xi_0$ from a point in $B_{\r_3}(\l(\xi_0)) \cap \Iloc$ (which we know exists by \Cref{lem:almost-valid}).
Since $\lstarss$ starts in steady state, $\lstarss(\theta) = \lstarss(\xi_0) + (\theta - \xi_0)\lambda$.
Consider now how the quantity $\lstarss_w(\theta) - \lstarss_v(\theta)$ varies with $\theta$.
If $e \in E^=$, then this is constant.
If $e \in E^>$, then this is increasing, and by \Cref{lem:Ebigger_and_Esmaller_in_ss}, $\lstarss_w(\xi_0) - \lstarss_v(\xi_0) > \tau_e$.
And similarly, if $e \in E^<$, then this quantity is decreasing and $\lstarss_w(\xi_0) - \lstarss_v(\xi_0) < \tau_e$.
In all cases, we deduce that $|\lstarss_w(\theta) - \lstarss_v(\theta) - \tau_e|$ is nondecreasing with $\theta$.
Thus, using that $\|\l(\xi_0) - \lstarss(\xi_0)\| \leq \r_3$ and the definition of $\xi_0$, 
\begin{equation}\label{eq:starss-far}
 |\lstarss_w(\xiend) - \lstarss_v(\xiend) - \tau_e| \geq |\lstarss_w(\xi_0) - \lstarss_v(\xi_0) - \tau_e| \geq |\l_w(\xi_0) - \l_v(\xi_0) - \tau_e| - 2\r_3 \geq \r_5 - 2\r_3.
\end{equation}
We also have
\[ \|\lstarind(\xi_0) - \lstarss(\xi_0)\| \leq \|\lstarind(\xi_0) - \l(\xi_0)\| + \|\l(\xi_0) - \lstarss(\xi_0)\| \leq \r_3 + \delta|V|, \]
and so $\|\lstarind(\xiend) - \lstarss(\xiend)\| \leq h(\r_3 + \delta|V|) = \r_4$.
This in turn implies that 
\[ \|\l(\xiend) - \lstarss(\xiend)\| \leq \r_4 + \epsilon_{j-1}. \]
Finally, combining this with \eqref{eq:starss-far}, we deduce that 
\[ |\l_w(\xiend) - \l_v(\xiend) - \tau_e| \geq \r_5 - 2\r_3 - 2\r_4 - 2\epsilon_{j-1} \geq 2K\delta + 3\epsilon_{j-1}. \]
It follows that $\xiend = \theta_1$, since otherwise by \Cref{lem:deltalip} we would be able to increase $\xiend$ while maintaining \eqref{eq:xiend}.
This completes the proof.
\end{proof}

\paragraph{Case 2: $\l$ is central near $\thetass$.}
Now we can focus on the key case.
Our first lemma gives us some explicit approximate comparisons between $\Delta \l / \Delta \theta$ and the steady-state direction $\lambda$; unlike \modified{\Cref{lem:similar_to_thin_flow}}, the comparisons are direct rather than involving the thin flow equations which define $\lambda$.
\begin{lemma}
\label{lem:hyperplane_dist_similar_to_thin_flow}
If $\l$ is central near $\thetass$ and $\delta$ is sufficiently small, then the following holds for any $\theta \in [\thetass + r_7, \theta_1]$:
\begin{enumerate}    
\renewcommand\labelenumi{(\roman{enumi})}
\renewcommand\theenumi\labelenumi
\item $\Delta \l_w - \Delta \l_v \leq (\lambda_w -\lambda_v) (\Delta \theta -r_7) + 3\r_5$ \quad for all $e = vw \in E^> \cup E^=$, \text{and} \label{it:E_bigger_arcs_dont_grow_queue_too_fast}
\item  $\Delta \l_w -\Delta \l_v \geq (\lambda_w -\lambda_v)(\Delta \theta - r_7) - 3\r_5$ \quad for all $e = vw \in E^< \cup E^=$. \label{it:E_smaller_arcs_dont_go_inactive_too_fast}
\end{enumerate}
\end{lemma}

\begin{proof}
The key idea of the proof of \ref{it:E_bigger_arcs_dont_grow_queue_too_fast} is that if $\l$ builds up a queue on some arc $e = vw \in E^>\cup E^=$ that is larger than a certain threshold, then $\l$ has a sufficiently large distance to the hyperplane $H_e$ that we can apply induction. 
In an equilibrium trajectory in steady state, the queue on such an arc grows at rate $\lambda_w - \lambda_v$;
the same then holds, approximately, for $\l$. \nnote{Why?}\laura{We used to call this argument close - middle - far in our discussions. I think what we want to say is: as soon as a queue grows, it grows with rate $\lambda_w - \lambda_v$ just by induction. This means that we can get an upper bound on the queue length. By case distinctions queues are not big in the beginning and as soon as they grow, we can control their growing rate. This is not enough, because it might happen that queues do not grow at all or start growing very late. For also getting a lower bound on the queue length we use the cut/robustness argument.}
Analogously the same holds for arcs $e \in E^<\cup E^=$ that become inactive by a margin. \nnote{TODO improve}

Assume \ref{it:E_bigger_arcs_dont_grow_queue_too_fast} is not true.  
Then there exists an arc $e = vw \in E^> \cup E^=$ such that together with the fact that $\ell$ is central near $\thetass$, we obtain a lower bound on the queue on that arc at time $\theta$.
\begin{align*}\l_w(\theta) - \l_v(\theta) - \tau_e &= \Delta \l_w - \Delta \l_v + \l_w(\thetass) - \l_v(\thetass) - \tau_e \\
&> (\lambda_w - \lambda_v) (\Delta \theta- r_7) + 3 \r_5 -\r_5  \\
&\geq 2\r_5.
\end{align*}
In other words, $\l(\theta)$ has a distance of at least $2\r_5 $ from the hyperplane $H_e$ while being on the side where $e$ is queueing.

Next, we define $\thetastart$ to be a point in time within $[\thetass + r_7, \theta]$ such that all $\xi \in [\thetastart, \theta]$ have a distance of at least $r_5$ from $H_e$ and $\dist(\l(\thetastart), H_e) \in [\r_5, \r_5 + \jump \delta]$. (Note that such a point exists due to \Cref{lem:deltalip} and because $\l$ is central near $\thetass$.) Moreover, this implies that during the whole interval there is at least a queue of length $\r_5$ on $e$.)
Hence, we can apply induction on $[\thetastart, \theta]$ by promoting $e$ to a free arc (i.e. setting $\Etildeind \coloneqq \Etilde$ and $\Einfind \coloneqq \Einf \cup \set{e}$). 
This provides us with an equilibrium trajectory $\lstarind$ which stays within distance $\epsilon_{j-1}$ of $\l$ on this interval.
Assuming that $\delta$ is sufficiently small, $\lstarind$ reaches steady state within time $1$, and so the requirement (ii) in \Cref{thm:technical} is satisfied (similarly to the argument in \Cref{lem:notcentral}).

\nnote{A remark for later, not for arxiv version: there seems to be a fair amount of repetition on these basic facts, between this proof and the previous lemma. The objects are not quite the same; but perhaps some efficiency gain is possible.} 
By \Cref{lem:staying_close_to_steady_state_region} $\l$ is within distance $\r_3$ of $\Iloc$ for the whole interval.
Hence there exists another equilibrium trajectory $\lstarss$ with $\lstarss(\thetastart) \in B_{\r_3}(\l(\thetastart)) \cap \Iloc$. 
Since $\lstarss$ starts in steady state,
\[\lstarss(\xi) = \lstarss(\thetastart) + \lambda (\xi - \thetastart) \quad \text { for all } \xi \in [\thetastart, \theta].\]
Since $\norm{\lstarss(\thetastart) - \lstarind(\thetastart)} \leq \r_3 + \abs{V}\delta$ by the definition of $\r_4$ \eqref{eq:r4}, we have
\[\norm{ \lstarss(\xi)- \lstarind(\xi) } \leq \r_4 \quad \text{ for all } \xi \in [\thetastart, \theta].\]
With the triangle inequality we obtain that
\[\norm{\l(\xi) - \lstarss(\xi)} \leq \r_4+ \epsilon_{j-1} \quad \text{ for all } \xi \in [\thetastart, \theta].\]

Putting everything together we obtain
\begin{align*}
	\Delta \l_w &- \Delta \l_v = \l_w(\theta) - \l_v(\theta) - \tau_e - (\l_w(\thetass) - \l_v(\thetass) - \tau_e) \\
	&\leq (\lstarss)_w(\theta) - (\lstarss)_v(\theta) - \tau_e + 2(\r_4 + \epsilon_{j-1}) + \r_5 \\
	&= (\lstarss)_w(\thetastart) + \lambda_w(\theta - \thetastart) - (\lstarss)_v(\thetastart) - \lambda_v(\theta - \thetastart)  - \tau_e + 2(\r_4 + \epsilon_{j-1}) + \r_5\\
	&\leq (\lambda_w - \lambda_v)(\theta - \thetastart) + \l_w(\thetastart) - \l_v(\thetastart)  - \tau_e + 2\r_3 + 2(\r_4+ \epsilon_{j-1}) + \r_5\\
	&\leq (\lambda_w - \lambda_v)(\theta - (\thetass +r_7)) + \r_5 + \jump \delta + 2\r_3 + 2(\r_4 + \epsilon_{j-1}) + \r_5\\
	&\underbrace{\leq}_{\eqref{eq:r5}} (\lambda_w - \lambda_v)(\Delta \theta -r_7) + 3\r_5.
\end{align*}
This is a contradiction to the assumption.

The proof of \ref{it:E_smaller_arcs_dont_go_inactive_too_fast} proceeds completely analogously. The only difference is that we remove the arc from the network instead of promoting it to be a free arc.
\end{proof}

%
    At this point we have built up enough information about $(\Delta \l, \Delta x)$\,---\,both in terms of some approximate thin-flow constraints that it satisfies, as well as some direct comparisons with the steady-state direction $\lambda$\,---\,that we are ready for the final part of the proof. 
    As already mentioned in the overview, it is inspired by the proof of Cominetti, Correa and Larr\'e~\CCL{} that the thin flow equations have a unique solution (in terms of the label derivative). 
    So we take the same approach: supposing for a contradiction that $\Delta \l$ is quite different from $\lambda \Delta \theta$, we can consider a nontrivial subset $S$, such that the ratio $\Delta \l_v / \lambda_v$ is significantly smaller for nodes $v \in S$ than it is for nodes $v \notin S$. 
    The control we have of $\Delta \l$ and $\Delta x$ from \Cref{lem:similar_to_thin_flow} and \Cref{lem:hyperplane_dist_similar_to_thin_flow} allow us to lowerbound the amount of the approximate circulation $\Delta x$ entering $S$, and upperbound the amount leaving, in a way that leads to a contradiction.

\begin{lemma}
\label{lem:quotients_close}
Given that $\l$ is central near $\thetass$ and $\delta$ sufficiently small,
\[\Delta \theta - \r_6 \leq \tfrac{\Delta \l_v (\theta)}{\lambda_v} \leq \Delta \theta + \r_6 \quad \text{ for all } \theta \in [\thetass + r_7 , \theta_1] \text{ and } v \in V.\]
\end{lemma}
\begin{proof}
Suppose for a contradiction that the claim does not hold.
So there exists a $\theta \in [\thetass + r_7 , \theta_1]$ and a node $v \in V$ such that $\frac{\Delta\l_v(\theta)}{\lambda_v} \notin [\Delta \theta - \r_6 , \Delta \theta + \r_6]$. 
Moreover,
\[\Delta \l_s(\theta) = \l_s(\theta) - \l_s(\thetass) = \theta - \thetass = \Delta \theta\] and $\lambda_s = 1$. Thus, $\frac{\Delta \l_s(\theta)}{\lambda_s} =\Delta\theta$.
Therefore, the distance between the quotients of $s$ and $v$ is at least $r_6$. Since there are at most $\abs{V}$ nodes, by the pigeonhole principle there exists a gap between two neighboring quotients of size at least $\gap \coloneqq \frac{\r_6}{\abs{V}}$. In other words there exists a threshold $z \in  [\Delta \theta - \r_6 , \Delta \theta + \r_6]$ such that we can partition $V$ into two non-empty parts
$S$ and $V \setminus S$ such that 
\[\max_{v \in S} \tfrac{\Delta \l_v}{\lambda_v} \leq z - \gap \quad \text{ and } \quad \min_{v \in V \setminus S} \tfrac{\Delta \l_v}{\lambda_v} \geq z.\]
The three claims that follow all relate to this set $S$, and hold within the context of our contradiction assumption.

\begin{claim}
\label{claim:arcs}
No arc in $E^= \cup E^>$ leaves $S$, and no arc in $E^= \cup E^<$ enters $S$.
\end{claim}
\begin{nestedproof}

%
%
%
%
%
Assume there is an arc $e=vw \in E^= \cup E^>$ with $v \in S$, $w\not\in S$. Then we have
\[
\Delta \l_v (\theta) \leq \lambda_v (z - \gap) < \lambda_v z 
\quad \text{ and } \quad \Delta \l_w (\theta) \geq \lambda_w z.
\]

Since $z \geq \Delta\theta - \r_6$ we obtain
\begin{equation} \label{eq:lowerbound_for_Ebigger}
\Delta \l_w (\theta) - \Delta \l_v (\theta) > (\lambda_w - \lambda_v) z \geq (\lambda_w -\lambda_v) \Delta \theta -\r_6 (\lambda_w -\lambda_v).
\end{equation}
We apply \Cref{lem:hyperplane_dist_similar_to_thin_flow} and obtain
\begin{equation} \label{eq:upperbound_for_Ebigger}
\begin{aligned}
\Delta\l_w & (\theta) - \Delta \l_v (\theta) \leq (\lambda_w -\lambda_v) \Delta \theta -\r_7 (\lambda_w -\lambda_v) + 3\r_5.
\end{aligned}
\end{equation}

Combining \eqref{eq:lowerbound_for_Ebigger} and \eqref{eq:upperbound_for_Ebigger} together with the definition of $\r_7$, we obtain 
\[ \r_6 (\lambda_w -\lambda_v) > (\lambda_w -\lambda_v) \Delta \theta - (\Delta \l_w (\theta) - \Delta \l_v (\theta)) \geq  \r_7 (\lambda_w - \lambda_v)  - 3 \r_5 \geq \r_6 (\lambda_w - \lambda_v), \]
which is a contradiction.

The argument for an arc $e=vw \in E^= \cup E^<$ entering $S$ is completely analogous. 
%
%
\end{nestedproof}

It will now be convenient to introduce a new arc $ts$ to the network. 
We will view $ts$ as a free arc of capacity $\nu_{ts} = u_0$ (the transit time can be any positive value, this does not matter). 
We
extend the thin flow $(\lambda, y)$ by setting $y_{ts} = u_0$ and $\Delta \lx_{ts} = \Delta \theta u_0$. This way $y$ is a
circulation and $\Delta \lx_{ts}$ is an approximate circulation in which the flow conservation differs by at most $2
\delta \nu_\Sigma$ (cf.\ \Cref{lem:similar_to_thin_flow}~\ref{it:Delta_x_flow}). We denote the arc sets of this extended
network with a bar. So $\Einfbar = \Einf \cup \set{ts}$  and moreover $\bar\delta^-(S)$ (respectively $\bar \delta^+(S)$)  denotes all arcs of $\Etilde \cup \{ts\}$ that enter (respectively leave) $S$.

The following is a technical claim, ruling out the possibility that $S$ defines a trivial cut despite being a nontrivial subset.
\begin{claim} \label{claim:leaving_arc_exists}
We have $\bar\delta^-(S) \neq \emptyset$.
\end{claim}
\begin{nestedproof}
Due to the definition of a thin flow, it follows that every node $v$ can be reached from $s$ via arcs in $\Etilde \setminus E^<$. (The path can be defined backwards by choosing an arc $e= vw$ that satisfies the second equation of \eqref{eq:thin} with equality).
This already yields the claim if $s \notin S$; take any $v \in S$ and observe that there is a path from $s$ to $v$, and hence an arc entering $S$.

If on the other hand $s \in S$, pick any $v \notin S$. There exists an $s$-$v$-path in $\Etilde \setminus E^<$, and hence an arc $e \in \Etilde \setminus E^<$ leaving $S$. By \Cref{claim:arcs}, $e$ cannot be in $E^= \cup E^>$, and so $e \in \Einfbar$. 
But all resetting arcs in a thin flow are flow-carrying (see \cite[Lemma 4.2]{OSV21}), and so $y_e > 0$. 
But since $y$ is a circulation, there must be a flow-carrying arc entering $S$ as well.
\end{nestedproof}

Next, we bound the net flow out of $S$ from below.

\begin{claim} \label{claim:net_outflow}
The following holds:
\begin{enumerate} \renewcommand{\labelenumi}{(\roman{enumi})}\renewcommand\theenumi\labelenumi
\item For all $e = vw \in \bar\delta^-(S)$ we have $\Delta \lx_e \leq z y_e - \gap \lambda_w \nu_e + \nu_e \delta$.
\item For all $e = vw \in \bar\delta^+(S)$ we have $\Delta \lx_e \geq z y_e - \nu_e \delta$.
\end{enumerate}
\end{claim}

\begin{nestedproof}
	First, we consider arc $ts$. We have $\Delta x_{ts} = u_0 \Delta \theta = \Delta \ell_s \nu_{ts}$ and moreover $y_{ts}=u_0$. If  $ts \in \bar\delta^-(S)$, then $s \in S$ and thus, $\Delta \l_s \leq \lambda_s (z -\gap) =(z -\gap)$ thus $(i)$ holds. If $ts \in \bar\delta^+(S)$, then $s \not\in S$, implying that $\Delta \l_s\geq \lambda_s z =z$ and thus $(ii)$ holds.

\begin{enumerate} \renewcommand{\labelenumi}{(\roman{enumi})}\renewcommand\theenumi\labelenumi
\item By \Cref{claim:arcs} we have $e \in E^> \cup \Einfbar$ and therefore \Cref{lem:similar_to_thin_flow} \ref{it:Delta_x_upper_bound} implies
\[\Delta \lx_e \leq \Delta\l_w \nu_e + \nu_e \delta \leq (z - \gap) \lambda_w \nu_e + \nu_e \delta = z y_e - \gap \lambda_w \nu_e +  \nu_e \delta.\]
The second inequality follows since $w \in S$ and hence $\Delta \l_w \leq (z - \gap) \lambda_w$. The last equality follows from the fact that $(y,\lambda)$ is a thin flow and $e \in E^> \cup \Einfbar$ implying $y_e = \lambda_w \nu_e$.

\item
By \Cref{claim:arcs} either $e \in \Einfbar$ or $e \in E^<$.
For $e \in \Einfbar$ we can combine \Cref{lem:similar_to_thin_flow} \ref{it:Delta_x_lower_bound}, $w \not\in S$, and the fact that $(y,\lambda)$ is a thin flow in which $e$ is resetting, to obtain:

\[\Delta \lx_e \geq \Delta\l_w \nu_e - \nu_e \delta \geq  z \lambda_w \nu_e - \nu_e \delta= z y_e - \nu_e \delta. \]

For $e \in E^<$ it simply holds that $y_e=0$ and thus,
\[\Delta \lx_e \geq 0 \geq z y_e - \nu_e \delta.
\]
\end{enumerate}
\end{nestedproof}

We are now ready to derive our desired contradiction.
The thin flow of the steady state $y$ is an $s$-$t$-flow of value $u_0$ and therefore a circulation in the extended network. 
However, we will see that $\Delta \lx$, which is an approximate circulation due to \Cref{lem:similar_to_thin_flow}~\ref{it:Delta_x_flow}, sends too much flow out of the cut compared to how much enters.

On the one hand, by \Cref{claim:net_outflow}, the fact that $y$ is a circulation, \Cref{claim:leaving_arc_exists}, and the definition of $\gap$ (respective the right term of $\r_6$) we have 
\begin{align*}
\sum_{e \in \bar\delta^+(S)} \Delta \lx_e - \sum_{e \in \bar\delta^-(S)} \Delta \lx_e 
&\geq \sum_{e \in \bar\delta^+(S)} z y_e - \delta \capsum   -   \sum_{e \in \bar\delta^-(S)} \left( z y_e - \gap \lambda_w \nu_e  \right) - \delta \capsum  \\ 
&=  \sum_{e \in \bar\delta^-(S)} \gap \lambda_w \nu_e  - 2 \delta \capsum \\
&\geq 3 \delta \nu_\Sigma \abs{V} - 2 \delta \capsum \\
& \geq 2 \delta \nu_\Sigma \abs{V}.
\end{align*}

On the other hand, by \Cref{lem:similar_to_thin_flow}~\ref{it:Delta_x_flow} it holds that
\[\sum_{e \in \bar\delta^+(S)} \Delta \lx_e - \sum_{e \in \bar\delta^-(S)} \Delta \lx_e 
= \sum_{v \in S} \left(\sum_{e \in \bar\delta^+(v)} \Delta \lx_e - \sum_{e \in \bar\delta^-(v)} \Delta \lx_e\right)
\leq 2 \delta \nu_\Sigma \abs{S} < 2 \delta \nu_\Sigma \abs{V}.\]
We have obtained our desired contradiction.
\end{proof}

Now we are ready to complete the proof of \Cref{thm:close_to_Nash}.  
\begin{proof}[Proof of \Cref{thm:close_to_Nash}]
The statement follows for $\theta \in[\thetass, \thetass +\r_7]$ by \Cref{lem:close_for_a_bit}.
So assume $\theta> \thetass + \r_7$.  If $\ell$ is not central near $\thetass$, then statement follows by \Cref{lem:notcentral}. 
So assume that $\ell$ is central near $\thetass$. 
%
%
Applying \Cref{lem:quotients_close} yields
\[\abs{\l_v (\theta) - (\l_v(\thetass) + \lambda_v \Delta \theta)} \leq \lambda_v\r_6 \quad \text{ for all } v \in V.\]

Since $\lss(\theta) = \lss(\thetass) + \lambda \Delta \theta$, we have that for all $v \in V$,
\begin{align*}
\abs{\l_v(\theta) - \lss_v(\theta)} &\leq  \abs{\l_v(\thetass) +\lambda_v \Delta \theta - \left(\lss_v(\thetass) + \lambda_v \Delta \theta \right)} + \lambda_v \r_6\\
                                      &\leq \r_3 + \lambda_v \r_6 \underbrace{\leq}_{\eqref{eq:r8}} \reight.
    \end{align*}
%
\end{proof}


\subsection{Putting it all together}

We can now complete the proof of the main technical theorem \Cref{thm:technical}; as already observed, \Cref{thm:main-full} follows as well.

\begin{proof}[Proof of \Cref{thm:technical}]
Define 
\[\epsilon_j \coloneqq \reight + h(\r_2 + \r_3),\]
where recall that $h(r)$ bounds the maximum distance two equilibrium trajectories starting within distance $r$ can ever be apart, with $h(r) \to 0$ as $r \to 0$.


    We have already seen that for $\theta \in [\theta_0, \thetass]$, $\|\lstar(\theta) - \l(\theta)\| \leq \r_2 \leq \epsilon_j$.
    So consider $\theta \in [\thetass, \theta_1]$.
    Let $\lss$ be an equilibrium trajectory starting from a point in $\Iloc$ as close as possible to $\l(\thetass)$.
    By \Cref{lem:staying_close_to_steady_state_region}, $\|\lss(\thetass) - \l(\thetass)\| \leq \r_3$, and so by
    \Cref{thm:close_to_Nash}, $\|\l(\theta) - \lss(\theta)\| \leq \reight$.
    But since $\|\lstar(\thetass) - \l(\thetass)\| \leq \r_2$, $\|\lstar(\thetass) - \lss(\thetass)\| \leq \r_2 + \r_3$.  
    Applying \Cref{thm:continuity_of_Nash}, it follows that $\|\lstar(\theta) - \lss(\theta)\| \leq h(\r_2 + \r_3)$, and hence that $\|\l(\theta) - \lstar(\theta)\| \leq \reight + h(\r_2+ \r_3) \leq \epsilon_j$.
\end{proof}

\section{Implications for approximate equilibria and packet models}\label{sec:applications}

\subsection{Approximate equilibria}

The goal is to show the following:

\epsilonequiisdeltatrajectory*

To do so we prove that $\epsilon$-equilibria satisfy an approximate Lipschitz-property; see \Cref{thm:epsilon_lipschitz}.
For this we start by showing that the potential mass that can overtake an agent $a$ is bounded. 
Recall that in an $\epsilon$-equilibrium, we have $d_t(a) \leq \l_t(\entry{a}) + \epsilon$ for all $a \in \agents$ where $\l$ and $d$ are the earliest arrival labels and departure time functions.

\begin{lemma}
\label{lem:not_too_much_flow_overtakes}
Let $\stratprof$ be an $\epsilon$-equilibrium for some $\epsilon > 0$.
Fix an agent $\a$. At most $3\epsilon\capsum$ particles can potentially overtake $\a$ at node $v \in P(\a)$, i.e., the measure of 
\[Q \coloneqq \Set{ \b \in \agents \mid \entry{\b} > \entry{a} \text{ and } \l_v(\entry{\b}) < \dep_v(\a) }\]
is bounded by $3\epsilon\capsum$. Recall that $\capsum = \sum_{e\in E} \nu_e + u_0$.
\end{lemma}

\begin{proof}
Consider the following set:
\[P \coloneqq \set{\b \in \agents | \dep_t(\b) \in [\l_t(\entry{\a}), \dep_t(\a) + \epsilon]}.\]
We have $Q \subseteq P$, by the following observations.
\begin{itemize}
\item No $\b \in Q$ can arrive before $\l_t(\entry{\a})$ by monotonicity of $\l$-labels.
\item No $\b \in Q$ can arrive after $\dep_t(\a) + \epsilon$. 
Observe that $\b$ can arrive
before $\dep_t(\a)$; i.e., $\l_t(\entry{\b}) \leq \dep_t(\a)$.
To see this, observe that $b$ could switch to path $\P(a)$, without any waiting at nodes ($b$ would be quicker at node $v$ than $\a$ then could follow the $v$-$t$ path of $\a$; while canceling cycles if needed).
Thus, by the $\epsilon$-equilibrium property the actual arrival time is at most $\dep_t(a) + \epsilon$.
\end{itemize}
Now, we aim to bound the mass of $P$.
As $\dep_t(\a) \leq \l_t(\entry{\a}) + \epsilon$, we have $[\l_t(\entry{a}), d_t(a) + \epsilon] \subseteq [\l_t(\entry{a}), \l_t(\entry{a}) + 2\epsilon]$.  Hence, particles with a mass of at most $2\epsilon \cdot \capsum$ can enter $t$ within this interval. Furthermore, particles with a mass of at most $\epsilon \cdot \capsum$ could have entered $t$ within $[\l_t(\entry{\a}) - \epsilon,\l_t(\entry{\a})]$ and waited at $t$ until $\l_t(\entry{\a})$ or later (since the maximal waiting time in $t$ is bounded by $\epsilon$.).  This implies that $\mu(Q) \leq \mu(P) \leq 3 \epsilon \capsum$.
\end{proof}

Similarly, the mass of flow that an agent can overtake is bounded by the same term.

\begin{lemma}\label{lem:bounded_overtaking}
Let $\stratprof$ be an $\epsilon$-equilibrium for some $\epsilon >0$.
Fix an agent $\a$. The flow mass that agent $\a$ can overtake at any $v \in P(\a)$ is bounded by $3\epsilon\capsum$, i.e., the measure of 
\[
R \coloneqq \Set{ \b \in \agents \mid \entry{\b} < \entry{\a}  \text{ and } \dep_v(\b) > \l_v(\entry{\a})}
\]
is bounded by $3\epsilon\capsum$.
\end{lemma}

\begin{proof}
Suppose that $\mu(R) > 3\epsilon\capsum$.
Look at some agent $\c \in R$ such that the measure of 
\[S \coloneqq \set{\b \in R \mid \entry{\b} > \entry{\c}}\]
is bigger than $ 3\epsilon\capsum$.
All agents that start between $\a$ and $\c$, i.e., after $\c$ and before $\a$ (so in particular all agents of set $S$) can overtake $\c$ by monotonicity of $\l_v$. 
This is a contradiction to \Cref{lem:not_too_much_flow_overtakes}.
\end{proof}

Finally, we prove that $\epsilon$-equilibria satisfy an approximate form of Lipschitz continuity.
\begin{theorem}
\label{thm:epsilon_lipschitz}
     There exist constants $\TrajLipDelta$ and $\jump$ (depending only on the instance) so that the following holds.
Let $\stratprof$ be an $\epsilon$-equilibrium for some $\epsilon >0$, and let $\l$ be the associated earliest arrival time labels.
Then for any vertex $v \in V$ and $\theta_2 > \theta_1$ we have
\[\l_v(\theta_2) \leq \l_v(\theta_1) + \TrajLipDelta (\theta_2 - \theta_1) + \jump \epsilon.\]
\end{theorem}
\begin{proof}
    We will prove the theorem with $\TrajLipDelta \coloneqq \TrajLip  \abs{V} = \max\{ 1, \tfrac{u_0}{\capmin}\} \cdot \abs{V}$ (recall $\TrajLip$ is the Lipschitz constant for equilibrium trajectories) and 
    $\jump \coloneqq  3 \TrajLipDelta \capsum$.

Let $P = (s\!=\!v_1, v_2, \ldots, v_k\!=\!v)$ be a shortest $s$-$v$-path, i.e., a path 
with $\l_{v_{i+1}}(\theta_1) = \l_{v_{i}}(\theta_1) + q_{e_i}(\theta_1) + \tau_{e_i}$, where $e_i = v_i v_{i+1}$ for all $i \in [k-1]$.
Inductively, we show that
\[\l_{v_i}(\theta_2) \leq \l_{v_i}(\theta_1) + C_i,\]
where $C_i \coloneqq  i\TrajLip(\theta_2 - \theta_1 + 3 \epsilon \capsum)$. 

For the induction base case, consider $i=1$. It holds that
\[\l_s(\theta_2) \leq \l_s(\theta_1) +(\theta_2 - \theta_1) <  \l_s(\theta_1) + \TrajLip(\theta_2 - \theta_1 + 3 \epsilon \capsum)=  \l_s(\theta_1) + C_1.\]
So assume the statement holds for $i$.
We have
\begin{align*}
\l_{v_{i+1}}(\theta_2) - \l_{v_{i+1}}(\theta_1) &\leq \l_{v_{i}}(\theta_2) - \l_{v_{i}}(\theta_1) + \lq_{e_i}(\theta_2) - \lq_{e_i}(\theta_1) +\tau_{e_i} - \tau_{e_i}\\
& \leq C_{i} + \tfrac{1}{\nu_e}\bigl(u_0(\theta_2 - \theta_1) + 3\epsilon\capsum\bigr)\\
&\leq C_{i+1}.
\end{align*}
To bound the difference of the flow mass in the queues in the second inequality, use that by \Cref{lem:bounded_overtaking} the mass of particles that could be overtaken by an (hypothetical) agent that enters the network at $\theta_1$ and travels without waiting along $P$ is bounded by $3\epsilon\nu_\Sigma$. Only this flow mass and the flow that entered the network between $\theta_1$ and $\theta_2$ can contribute to the difference of the queues.

As $P$ is a simple path we have $k \leq \abs{V}$, which finishes the proof.
\end{proof}

Finally, we can prove that every $\epsilon$-equilibrium is a $\delta$-trajectory for appropriate~$\delta$.

\begin{proof}[Proof of \Cref{thm:epsilon_equi_is_delta_trajectory}]
Choose 
\[\delta \coloneqq 3 \epsilon  \TrajLipDelta \tfrac{\capsum}{u_0} + 3 \epsilon \TrajLipDelta \capsum\]
and consider an agent $\a$ and a node $v$.
For $\theta_2 \coloneqq \entry{\a} + 3 \epsilon \tfrac{\capsum}{u_0}$ it holds that the total mass of flow entering the network within $[\entry{\a}, \theta_2]$ is given by $3 \epsilon \capsum$.
Due to \Cref{lem:not_too_much_flow_overtakes} and the monotonicity of $\l_v$ we have $d_v(\a) \leq \l_v(\theta_2)$ since otherwise the measure of overtaking particles would be strictly larger than $3 \epsilon \capsum$.

Finally, using the approximate Lipschitzness derived in \Cref{thm:epsilon_lipschitz}, we obtain
\begin{align*}
\dep_v(\a) &\leq \l_v(\theta_2) \\
&\leq \l_v(\entry{\a})  + \TrajLipDelta (\theta_2 - \entry{\a}) + \jump \epsilon\\
&= \l_v(\entry{\a})  + \TrajLipDelta \cdot 3 \epsilon   \capsum \tfrac{1}{u_0} + 3 \epsilon \TrajLipDelta \capsum \\
& = \l_v(\entry{\a}) + \delta.
\end{align*}
\end{proof}

\subsection{Packet models}
\label{subsec:packet_model}



We consider a packet model inspired by the one used in \cite{hoefer2011competitive}.
Given a network, where very arc~$e$ has a transit time $\tau_e > 0$ and a capacity $\nu_e > 0$. 
Moreover, give a set of infinitely many packets (indexed by $\N$), each of size $\beta$. 
(While we assume all packets to have the same size, it would be straightforward to extend the results to a model where $\beta$ is an upper bound on the size of any packet.)
The packets start strictly ordered in front of the source $s$ and enter then one by one with a given network inflow rate $u_0$ into the network. In other words, packet $i \in \N$ leaves $s$ at time $\beta \cdot i / u_0$.
Note that this is a small deviation from the model in \cite{hoefer2011competitive} where all packets start at time $0$. 
For the transition from their model to ours one can imagine adding an artificial source $\hat s$ and arc $\hat s s$ with a capacity of $u_0$.

The goal of each packet is to reach a common sink $t$ as early as possible. 
A strategy profile in this game is a simple $s$-$t$ path for every packet.

Given a strategy profile, a packet travels along paths by traveling along the arcs of the path in the appropriate order. 
For each arc the packet first needs to be processed for $\beta/\nu_e$ time units. 
If a previous packet is still being processed when a new packet arrives, the new packet joins a waiting queue. The ordering in the queue is with respect to entrance time into the arc.
If multiple packets enter the arc at exactly the same time, tie-breaking is based on packet index; the packet that departed the source first has priority.
After being processed, a packet traverses the arc for $\tau_e$ time units (note that this is different from \cite{hoefer2011competitive} where $\tau_e=0$ for all edges).
It can then start to traverse the next arc of the path. 
Note that time is not discrete in this model; packets may enter or leave an arc at arbitrary moments in time.

A strategy profile describes an equilibrium if no packet can change its path and improve its arrival time at the sink $t$.

\paragraph{Expressing the packet model in our flow over time model.}
In the following we describe how to define for every packet model strategy profile a corresponding strategy profile in our flow over time model.
Each packet is replaced by an infinite collection of agents with a total measure of $\beta$, and each of these agents will take the same path as the packet did. 
A bit more care will be needed in defining waiting times for the agents, in order to match precisely the behavior of the packet model.

We say that all agents $a \in \agents$ with $\entry{a} \in ((i-1)\beta, i\beta]$ belong to packet $i$, for $i \in \N$. In additional all agents with $\entry{a} = 0$ belong to packet $1$.

Given some strategy profile $\packetprof$ in the packet model, the corresponding strategy profile $\stratprof = (\P, \Wait)$ in the flow over time model is defined as follows.
We define $\P(a)$ to be the path of packet $i$ in $\packetprof$, where $i$ indexes the packet that $a$ belongs to.
We view agent $a$ as having a position in the packet, given by $x_{a} = \entry{a} - (i-1)\beta$.
Then for any arc $e=uv \in \P(a)$, we define $\Wait(a)_v = (\beta-x_{a}) \tfrac{1}{\nu_e}$; and furthermore, we define $\Wait(a)_s = (\beta-x_{a}) \tfrac{1}{u_0}$.

With this definition, all agents belonging to some packet $i$ enter an arc at the same time, and at the same time as the packet entered the arc. 
In other words the departure time of an agent coincides with the departure time of the packet it belongs to. 
At $s$ this is immediate: waiting times are chosen precisely so that departure times correspond.
For other nodes, it follows inductively based on the departure time of a packet.
Suppose the correspondence holds for all packets departing from any node strictly before some time $\xi$, and consider a packet $i$ that takes an arc $uv$ and departs $v$ at time $\xi$.
All the constituent particles of the packet departed from $u$ in $\stratprof$ at precisely the time that packet $i$ departed from $u$.
These particles will be delayed by particles of earlier packets potentially, and then stream out of the arc in an interval of time of length precisely $\beta/\nu_e$. 
The arrival time of the last particle in the packet will be precisely as intended, and earlier particles will wait precisely the correct amount of time.
\pagebreak[4]
\paragraph{Packet equilibria as strict $\delta$-equilibria.}
\packets*

\begin{lemma}
\label{lem:common_lemma}
Consider some 
strategy profile $\sigma$ in the packet model, with corresponding flow-over-time strategy profile $\stratprof$.
Let $\quickdep_v(\a)$ be the hypothetical departure time from $v \in \Path(\a)$ if $\a$ were to maintain its path $\Path(\a)$, but \emph{not} wait at any node. 
Then for any agent $\a$, 
\[
    \dep_t(\a) \leq \quickdep_t(\a) + O(\beta).
\]
\end{lemma}
\begin{proof}
    Let $\epsilon := \beta / \capmin$ in what follows. Note that $\epsilon$ is an upper bound on the time that any particle waits at any node, just from the description of the way $\stratprof$ is constructed.

We argue along the nodes $s=v_1, v_2, \dots, v_k =t$ of path $P := \Path(\a)$. We show that while the difference between $\dep$ and $\quickdep$ may increase along the path, we can bound the increase.

\begin{claim}
\label{claim:error_along_path}
For each $i$,
\[\dep_{v_i} (\a ) \leq \quickdep_{v_i}(\a) + C_i,\]
where $C_1 \coloneqq \epsilon$ and $C_i \coloneqq \frac{(C_{i-1} + \epsilon) \cdot \sum_{e \in \delta^-(v_{i-1})} \nu_e }{\nu_{v_{i-1} v_i}} + C_{i-1}$ for all $i \geq 2$. 
\end{claim}
\begin{nestedproof}
    Since the waiting time of each agent at each node is bounded by $\epsilon$, this applies to the source $s = v_1$ and thus $\dep_{v_1}(\a) \leq \quickdep_{v_1}(\a) + \epsilon.$

\medskip

Now suppose the claim holds for some $i$. 
We bound the delay $\a$ experiences at $v_{i+1}$ compared to traveling without waiting by the additional delay the agent gets compared to the delay at $v_i$, by bounding the mass of flow on arc $v_iv_{i+1}$ that delays agent $\a$, but not a probe particle traveling along $P$ without waiting. 
Any flow departing $v_i$ within the interval $[\quickdep_{v_i}(\a), \dep_{v_i}(\a)]$ must have arrived at $v_i$ within the interval
$[\quickdep_{v_i}(\a)-\epsilon, \dep_{v_i}(\a)]$, since waiting times are bounded by $\epsilon$.
The mass of flow that can arrive at $v_i$ within the interval $[\quickdep_{v_i}(\a) - \epsilon, \dep_{v_i}(\a)]$, which has length at most $C_i + \epsilon$, 
is bounded by $(C_i + \epsilon) \sum_{e \in \delta^-(v_i)} \nu_e$.
Any increase in the value of $\dep_{v_{i+1}}(\a) - \quickdep_{v_{i+1}}(\a)$ compared to $\dep_{v_i}(\a) - \quickdep_{v_i}(\a) \leq C_i$ is due to this flow, which can cause an additional delay of at most
\[
    \frac{(C_i + \epsilon) \cdot \sum_{e \in \delta^-(v_i)} \nu_e }{\nu_{v_i v_{i+1}}};
\]
The claim follows.
\end{nestedproof}

\begin{claim} \label{claim:simple_bound}
For all $i \in \set{1, \dots, k}$:
\[C_i \leq \epsilon i 2^{i} \left( \tfrac{ \nu_{\Sigma} } {\nu_{\min}} \right)^{i-1}.\]
\end{claim}
\begin{nestedproof}
The claim is immediate for $i=1$, so suppose $i > 1$.
We proceed by induction, so assume the claim holds for $i-1$.
By the definition of $C_i$ in \Cref{claim:error_along_path}, along with the fact that $\frac{\nu_\Sigma}{\nu_{\min}} > 1$, we have
\[
C_{i} \leq ( 2 C_{i-1} + \epsilon) \tfrac{\nu_\Sigma}{\nu_{\min}}. 
\]
A straightforward calculation completes the proof:
\[ C_{i} \leq ( 2 C_{i-1} + \epsilon) \tfrac{ \nu_\Sigma  }{\nu_{\min}} 
\leq 2 \epsilon(i-1)2^{i-1} \left( \tfrac{ \nu_{\Sigma} } {\nu_{\min}} \right)^{i-1}  + \epsilon \tfrac{\nu_{\Sigma}}{\nu_{\min}}  
< \epsilon i 2^{i} \left( \tfrac{ \nu_{\Sigma} } {\nu_{\min}} \right)^{i-1}.
\]
\end{nestedproof}

Using \Cref{claim:simple_bound} with the fact that $P$ is a simple path and hence has length bounded by $\abs{V}$ completes the proof of the lemma.
\end{proof}

\begin{proof}[Proof of \Cref{thm:packets}]

We fix some packet $p$. Let $\a_{\ell}$ be a last agent (i.e., an agent with maximal entrance time into the network) of that packet and $\a_f$ be a first agent. 
Let $Q$ be the path packet $p$ takes.
Let $P$ be an earliest arrival path for $\a_{\ell}$, i.e., a path determining $\l_t(\entry{\a_{\ell}})$.
(Note that this may differ from the path $\Path(\a_{\ell})$ that the packet $p$ takes, since $P$ is a best path if agent $\a_{\ell}$ unilaterally deviates, not if the entire packet $p$ deviates.)

Let $\sigma'$ be the strategy profile for the packet model where packet $p$ chooses path $P$ instead of path $Q$ (with all other packet strategies unchanged).
Let $\stratprof'$ be the flow-over-time strategy profile corresponding to $\sigma'$.

The agent $\a_{\ell}$ can be seen as the decision maker of packet $p$.
It considers all possible $s$-$t$-paths $P'$ and simulates her personal arrival time if the full packet traverses this path. As the earlier agents of the packet always wait for the last agents on each node, the arrival time of the packet, as well as the arrival time of every agent in the packet, coincide with the arrival time of $\a_{\ell}$. Since we consider an equilibrium in the packet model, $Q$ is a path for which this arrival time is as early as possible.
Thus
\[
\dep^{\stratprof}_t(\a_{\ell}) \leq \dep_t^{\stratprof'}(\a_{\ell}) = \dep^{\stratprof'}_t(\a_f) .
\]

We now bound $\dep_t^{\stratprof'}(\a_f)$ via \Cref{lem:common_lemma}. 
With $\quickdep$ as in the lemma statement (for strategy profile $\stratprof'$), we have that 
$\dep^{\stratprof'}_t(\a_f) \leq \quickdep_t(\a_f) + O(\beta)$.
Moreover, we get  $\quickdep_t(a_f) \leq \l_t(\entry{a_\ell})$ since $\entry{a_f}\leq \entry{a_\ell}$, by choice of $P$ and monotonicity of the $\ell$-labels.
So altogether,
\begin{equation}\label{eq:last_packet}
\dep^{\stratprof}_t(\a_\l) \leq \dep^{\stratprof'}_t(\a_\l) = \dep^{\stratprof'}_t(\a_f)\leq \quickdep_t(\a_\l) + O(\beta) \leq \l_t(\entry{a_\l}) + O(\beta). 
\end{equation}
This gives us the desired bound, but only for the last agent of a packet.

We now aim to obtain the bound for every agent of a packet. 
Consider an arbitrary agent $a$, and let $p_1$ be the corresponding packet. 
For now, suppose that $p_1$ is not the first packet, and let $p_0$ be the previous packet.
Let $a_0$ and $a_1$ be the last agents of $p_0$ and $p_1$.

\begin{claim}
\label{claim:packet_following}
\[\dep_t(\a_1) \leq \dep_t(\a_0) + \tfrac{\beta}{\capmin}.\]
\end{claim}
\begin{nestedproof}
    As before we use the fact that on a shortest path $p_1$ will not be delayed by subsequent agents (neither in the packet model nor in the constructed flow model). Thus, to obtain an upper bound on the shortest path travel time, we can as well consider the truncated game where  after $p_1$ no packet enters into the network, i.e., no agent enters after time $\entry{a_1}$.
    One option for $p_1$ to reach $t$ is to follow the same path $P_0$ as $p_0$.
    We claim that the arrival time of this strategy is upper bounded by $\dep_t(a_0) + \tfrac{\beta}{\capmin}$ in the truncated game. From this fact the claim follows, since the actual path $p_1$ uses can only lead to an earlier arrival time given equilibrium behavior.

     We note that when following $p_0$, no other packets will enter an arc after $p_0$ but before $p_1$, since we truncated the game. 
%
     Considering some arc $e=vw$ on the path $P_0$, the arrival time of the first agent of $p_1$ at $w$ is at least the arrival time of the last agent of $p_0$ at $w$. Let us say that $p_0$ and $p_1$ ``touch'' at $w$ if these two times are in fact equal.
     Let $v$ be the last node on path $P_0$ where $p_0$ and $p_1$ touch (they certainly touch at $s$, so such a node exists).
     Then $p_1$ departs from $v$ precisely $\frac{\beta}{\nu_e} \leq \frac{\beta}{\capmin}$ later than $p_0$.
     Between $v$ and $t$ on $P_0$, packet $p_1$ is never delayed\,---\,the previous packet $p_0$ has been processed by the time that $p_1$ arrives.
     It follows that $p_1$ departs $t$ at most $\frac{\beta}{\capmin}$ later than $p_0$, as required.
 
 %
%
\end{nestedproof}


We now have 
    \begin{align*}
    d_t(\a) &= d_t(\a_1) && \text{(all agents of packet $p_1$ depart at the same time)} \\
            &\leq d_t(\a_0) + \tfrac{\beta}{\capmin} && \text{(by \Cref{claim:packet_following})}\\  
            &\leq \l_t(\entry{\a_0}) + \tfrac{\beta}{\nu_{\min}}+ O(\beta) &&\text{(by \eqref{eq:last_packet} for $\a_0$)}\\ 
            &\leq \l_t(\entry{\a}) + \tfrac{\beta}{\nu_{\min}} + O(\beta) &&\text{(monotonicity of $\l$)}\\
            &=  \l_t(\entry{\a}) +  O(\beta).
    \end{align*}


Lastly, we need to consider the case where $p_1$ is the first packet.
The argument is essentially the same, except that we should consider $a_0$ to be a hypothetical single agent that takes an earliest arrival path to $t$ in the empty network. (Since this is just a single nonatomic agent, including this has no impact whatsoever.)
The statement of \Cref{claim:packet_following} is easily seen to still hold, as does the sequence of inequalities above.

%
%

So $\stratprof$ is an $O(\beta)$-equilibrium. By \Cref{thm:epsilon_equi_is_delta_trajectory} this implies that $\stratprof$ is also a strict $\delta$-equilibrium for some $\delta = O(\beta)$, completing the proof.
\end{proof}

\section{Conclusion}\label{sec:conclusion}

We have demonstrated that strict $\delta$-equilibria converge to exact dynamic equilibria in the deterministic queueing model, 
and as two specific consequences, derived the convergence of $\epsilon$-equilibria and of equilibria in a specific packet-routing model.
But we emphasize that these are merely two consequences, and others can surely be obtained.
Convergence of other packet models can certainly be demonstrated, as well as stability with respect to small perturbations of network parameters such as transit times and capacities.
We leave detailed investigation of this to future work.

It must be admitted that our bounds are not very effective: we do not explicitly compute how $\epsilon$ depends on $\delta$, but our dependence is certainly (at least) exponential, and very dependent on the specific network being considered. 
This issue arises even in the setting of continuity. One could hope that the correct dependence is much better, perhaps linear. 
But it is not clear how this can be approached with the current techniques of this paper.
A related issue already mentioned is that our results only apply for $\delta$ sufficiently small (where ``sufficiently small'' depends on the instance). 
This seems like a potentially much easier issue to resolve.

As already mentioned, our result potentially allows for the transfer of results from the deterministic queueing model to atomic packet-routing models.
This would especially be the case if the restriction to sufficiently small $\delta$ can be removed.
There may be further applications of this, now or in the future, as a better understanding of the deterministic queueing model is obtained.

\paragraph*{Acknowledgements.}
We are grateful to an anonymous referee for an extremely careful reading, and their helpful and detailed remarks.

We thank the Hausdorff Trimester Programme on Discrete Optimization, where some part of this work was done.
We also acknowledge the Dagstuhl seminar series on Dynamic Traffic Models in Transportation, which has been a very inspiring environment for discussing the topic generally.

\bibliographystyle{alpha}
\bibliography{literature}

\appendix

\section{Omitted proofs}\label{app:proofs}


\loading*

\begin{proof}
We will iteratively construct sets $A^+_e(\xi)$ and $A^-_e(\xi)$, for all arcs $e$ and times $\xi$, which will represent the set of agents that enter or exit arc $e$ by time $\xi$.
So $F^+_e(\xi) = \mu(A^+_e(\xi))$ and $F^-_e(\xi) = \mu(A^-_e(\xi))$ for all $e, \xi$.
The departure time functions are also fully determined by $A^+$ and $A^-$: for any agent $a$ and $e=uv \in \P(a)$,
\begin{equation}\label{eq:Atod}
    \begin{aligned}
        d_u(a) &= \inf\{\xi: a \in A^+_e(\xi)\}, \text{ and }\\ 
    d_v(a) &= \inf\{ \xi : a \in A^-_e(\xi)\} + \Wait_v(a).
    \end{aligned}
\end{equation}
(Except at $s$ and $t$, either equation can be used to determine $d_v(a)$ for $v \in \P(a)$.)

Clearly $A^+_e(0)$ and $A^-_e(0)$ exist, are unique, and are easy to determine (most are empty, except possibly for $A^+_e(0)$ for $e \in \delta^+(s)$).
Suppose we have determined $A^+_e(\xi)$ and $A^-_e(\xi)$ for all $e$ and for all $\xi \leq \xi_0$, for some $\xi_0$.
This determines $F^+_e(\xi)$ and $F^-_e(\xi)$, and hence also the queue volumes $z_e(\xi)$ (via \eqref{eq:queues_operate_at_capacity}), for all $e$ and $\xi \leq \xi_0$.
It also determines the value of $d_v(a)$ for any $a \in \agents$ and $v \in \P(a)$ for which $d_v(a) \leq \xi_0$; this follows from \eqref{eq:Atod}.
Finally, this determines $\dq_e(a)$ for all agents $a$ and arcs $e=vw \in \P(a)$ in which $d_v(a) \leq \xi_0$, by \eqref{eq:queue-waiting}.

We will show how to uniquely determine the values of $A^+_e(\xi)$ and $A^-_e(\xi)$ for all arcs $e$ and $\xi \leq \xi_1$, for $\xi_1 = \xi_0 + \min_e \tau_e > \xi_0$.
This implies that $A^+$ and $A^-$ can be fully specified for all times, and moreover are unique.

First, we can easily determine $A^-_e(\xi)$ for any $e=vw$ and $\xi \leq \xi_1$: 
\[ A^-_e(\xi) = \{ a \in A^+_e(\xi-\tau_e)  \text{ and } d_v(a) + \dq(a) + \tau_e \leq \xi\}.
\]

Since every agent that aims to leave the arc by time $\xi$ has to enter the arc at the latest at time $\xi -\tau_e$, we can restrict our attention to the agents in $A^+_e(\xi-\tau_e)$. This set and also the $d_v$ values corresponding to its agents are already known at this point, since $\xi-\tau_e \leq \xi_0$.

Once $A^-$ has been determined until time $\xi_1$, we can determine $A^+_e(\xi)$ for $e=vw$ and $\xi \leq \xi_1$ via
\[
    A^+_e(\xi) = \{ a \in \agents: e \in \P(a) \text{ and } \exists f \in \delta^-(v) \text{ s.t. } a \in A^+_f(\xi - \Wait_v(a)) \}
\]
if $v \neq s$, and
\[ 
    A^+_e(\xi) = \{ a \in \agents: e \in \P(a) \text{ and } \entry{a} + \Wait_s(a) \leq \xi \}
\]
if $v=s$. 
\end{proof}


\steadystate*

\begin{proof}
%

	We use  \cite[Theorem 4.14]{OSV21} to describe the time the equilibrium needs to reach steady-state. 
	We will use notation from the statement of this theorem. 
	
	The first part is the time until the trajectory reaches pre-steady-state which is $T_1=\frac{\OPT - \Phi(0)}{\eta}$. 
	For each point $\linit \in I$ the potential equals its maximal value $\OPT$ from time 0 on. Moreover the potential is Lipschitz continuous. Let us denote the Lipschitz constant as $C_{Pot}$ and $r \coloneqq d(\lcirc,I)$. Then $\OPT -\Phi(0) \leq C_{Pot} \cdot r$ and thus $T_1 \leq \frac{C_{Pot} \cdot r}{\eta}$.
	
	To get some bound on $T_2$, remember that $s_e(0)$ denotes the slack $[\linit_v + \tau_e -
	\linit_w]^+$ for an arc $e = vw \in E \setminus \Einf$ (and $s_e(0) = 0$ for $e \in \Einf$), i.e., it measures how
	inactive that arc is. For $\linit \in I$ we have $s_e(0) = 0$ for all $e \in E \setminus E^<$. This implies that for a point in distance $r$ of $I$ the maximal slack on some arc of $E \setminus E^<$ is upper bounded by $2r$.
	Using the above bound for $T_1$, we obtain the following bound on $T_2$.
	
	\[
	T_2 \leq \Delta \abs{E} \kappa \frac{ C_{Pot}\cdot r}{\eta}  + \Delta \abs{E} 2r.
	\]
	
	Thus with choosing $\T \coloneqq (\Delta \abs{E} \kappa +1 ) C_{Pot}/ \eta + \Delta \abs{E} 2$, we obtain the wished result. 
\end{proof}

\hyperplaneseparation*
\begin{proof}

%
    \emph{(i).} Consider some compatible set $F$. It suffices to show that there is a constant $\centerdist_F$ so that $\dist(\linit, H_F) \leq \centerdist_F \max_{e \in F} \dist(\linit, H_e)$  for all $\linit $, since there are only a finite number of choices of $F$.

    Since all norms of a finite dimensional vector space are equivalent, it suffices (up to constant factors) to show this claim for a different norm, which we define as follows.
    Let $W$ be a linear subspace and $p \in \R^V$ be such that $H_F = W + p$.
    Let $n_e$ denote the unit (with respect to the 2-norm) normal to $H_e$, for each $e \in F$.
    Then $W^\perp$ is spanned by $\{ n_e : e \in F\}$, and in particular, $\{ n_e : e \in F'\}$ is a basis for $W^\perp$ for some $F' \subseteq F$.
    Thus we can define a basis $(q_i)_{i \in V}$ for $\R^V$ consisting of $\{ n_e : e \in F'\}$ followed by $|V| - |F'|$ other vectors forming a basis for $W$.
    Now consider the norm defined by 
    \[ \|x\|_* := \sum_{i} |\alpha_i| \qquad \text{ for } \qquad x = \sum_{i} \alpha_i q_i.  
    \]
    Let $d_*(x, S)$ denote the minimum distance from a point $x$ to a set $S$ with respect to $\|\cdot\|_*$.
    Write $\linit - p = \sum_i \beta_i q_i$ for appropriate $\beta_i$'s.
    Then 
    \[ 
        d_*(\linit, H_F) = d_*(\linit - p, W) = \sum_{i=1}^{|F'|} |\beta_i| = \sum_{e \in F'} d_*(\linit, H_e) \leq \sum_{e \in F} d_*(\linit, H_e). 
    \]
    So $d_*(\linit, H_F) \leq |F|\max_{e \in F} d_*(\linit, H_e)$. Since $\|\cdot\|_*$ is equivalent to $\|\cdot\|$, the existence of the desired constant $\centerdist_F$ follows.

    \medskip

    \emph{(ii).}
Let 
\[ \sigma \coloneqq \min \{ \dist(H_e, H_F) \mid F \text{ is a compatible set not containing } e\}. \]
This is well defined since the minimum is taken over a finite set. Moreover, every such distance is positive by definition, so $\sigma > 0$.
Now let $\centercompat := \sigma / (1 + \centerdist)$.

    Let $F$ be the set of hyperplanes intersecting $B_{\centercompat}(\linit)$.
    Suppose $F$ was not compatible; choose $F' \subseteq F$ maximal so that $F'$ is compatible, and $e \in F \setminus F'$.
    Then $\dist(\linit, H_{F'}) \leq \centerdist \centercompat$ by (i).
    But now 
    \[ \dist(H_e, H_{F'}) \leq \dist(\linit, H_e) + \dist(\linit, H_{F'}) < (1+\centerdist)\centercompat = \sigma, \]
    contradicting our choice of $\sigma$.
\end{proof}

\almostvalid*

\begin{proof}
    Property (i) of a valid label, that there is a path from $s$ to every node $v$ in $E'_{\l(\theta)}$, holds by the definition of earliest arrival labels. Property (iii) follows since for any vector $\lclose \in \Rplus^V$ the set of active arcs is acyclic by our assumption that $\tau_e > 0$ for all arcs $e$; if we sum $\lclose_w - \lclose_v - \tau_e$ over all arcs in a directed cycle, the result is negative, and so not all of these values can be nonnegative. But this a contradiction since there is no free arc on the cycle as $(E, \Einf)$ is a valid configuration. 
    So the only issue is with property (ii), that for every arc $e=vw \in E^*_{\l(\theta)}$ there should be a path from $w$ to $t$ in $E'_{\l(\theta)}$. This may not hold.

We will define some $\lclose \in \R^V$ with $\l_v(\theta) -|V|\delta \leq \lclose_v \leq \l_v(\theta)$, and show that $\lclose \in \Omega$.
We will do this by initially setting $\lclose = \l(\theta)$, and gradually modifying it, maintaining property (i), and making progress on property (ii).

Given our current choice of $\lclose$, define $Q$ to be the set of nodes $w$ for which there is no $w$-$t$-path in $E'_{\lclose}$, but $w$ is the head of some arc of $E^*_{\lclose}$.
If $Q$ is empty we are done.
Otherwise, let $T$ denote the set of nodes reachable from $Q$ via arcs of $E'_{\lclose}$.
Now decrease, uniformly, all labels in $T$, that is, replace $\lclose$ with $\lclose - \epsilon \chi(T)$.
Here, $\epsilon$ is chosen small enough that no arcs leaving $T$ join $E^*_{\lclose}$, and all nodes in $Q$ still have an active entering arc.
(In other words, we decrease the labels in $T$ until either $Q$ or $T$ would change.)
Then recompute $Q$ and $T$, and repeat this process, until $Q$ is empty.

\medskip

First, we claim that property (i) is maintained in this process.
Let $Q$, $T$ be as determined before the update, and let $\lclose$ and $\lcloseup$ be the old and new labelings.
When we refer to ``active'' arcs in what follows, we mean the arcs of $E'_{\lclose}$.
By the definition of $T$, an active path can enter $T$, but cannot leave.
To argue about nodes in $T$, first let $Q'$ be a minimal subset of $Q$ with the property that everything in $T$ is reachable from $Q'$ via active arcs.
Since $E'_{\lclose}$ is acyclic, $Q'$ is unique, and simply consists of all nodes in $Q$ that cannot be reached via active arcs from another node of $Q$.
Every node $w \in Q'$ must have an incoming arc $vw \in E^*_{\lclose}$; further, $v \notin T$, since there can be no active path from $Q$ to $v$ by our choice of $Q'$.
Since $vw \in E^*_{\lclose}$, $vw \in E'_{\lcloseup}$, so $w$ is still reachable from $s$ along arcs in $E'_{\lcloseup}$.
Since this holds for all nodes in $Q'$, it holds for all of $T$.

\medskip

Finally, we show that this process terminates quickly with $Q$ becoming empty, at which point property (ii) is satisfied.
Define the \emph{slack} of a node $w$ in $Q$ with respect to a labeling $\lclose$ to be the minimum slack of a $w$-$t$-path, where the slack of an arc $e=uz$ is $0$ if it is active, and $-(\lclose_z - \lclose_u - \tau_e)$ if it is inactive.
Then in each step where we decrease the labels of $T$ by $\epsilon$, all nodes currently in $Q$ have their slack decreased by $\epsilon$.
So it suffices to show that at the start, each node in $Q$ has slack at most $|V|\delta$.

So consider the situation at the start, with $\lclose = \l(\theta)$, and pick some $w \in Q$. Let $e=vw$ be the arc entering $w$ that is in $E^*_{\lclose}$.
For any $\epsilon$, we can choose an agent $a$ with $e \in \Path(a)$, $d_v(a) \leq \l_v(\theta)$, and where the measure of agents $a'$ who use arc $e$ and depart $v$ in the interval $[d_v(a), \l_v(\theta)]$ and have $\entry{a'} < \entry{a}$, is less than $\epsilon$.
Let $P$ be the subpath of $\Path(a)$ from $w$ to $t$.

Now fix any arc $e'=uz$ of $P$.
We have $d_z(a) \leq \l_z(\theta) + \delta$ because it is a $\delta$-trajectory, and $\entry{a} \leq \theta$.
We also have $d_u(a) \geq \l_u(\theta) - O(\epsilon)$.
This is because it is almost at the end of the queue for $e$ at time $\l_v(\theta)$ (only mass $\epsilon$ can be behind it), and so it arrives at $w$ no earlier than $\l_v(\theta) + \lq_e(\theta) + \tau_e - O(\epsilon) \geq \l_w(\theta) - O(\epsilon)$.
Similarly for all later nodes on the path $P$, including $u$.

Since $d_u(a) \geq \l_u(\theta) - O(\epsilon)$ and $a$ uses arc $uz$, we have that $d_z(a) \geq \l_u(\theta) + \lq_{e'}(\theta) + \tau_{e'} - O(\epsilon)$. 
We conclude that $\l_z(\theta) - \l_u(\theta) - \lq_{e'}(\theta) - \tau_{e'} \geq -\delta - O(\epsilon)$.
Taking $\epsilon$ to $0$, we get that the slack of $uz$ is at most $\delta$.
Since $P$ contains less than $|V|$ arcs, this path demonstrates that the slack of $w$ is less than $|V|\delta$.
\end{proof}

\end{document}